\documentclass[final,leqno]{siamltex}
\usepackage{fullpage}
\usepackage{amsmath}
\usepackage{hyperref}
\usepackage{graphicx}
\usepackage{placeins}
\usepackage{amssymb}
\usepackage{caption}
\usepackage{subcaption}
\usepackage{algorithmicx,algpseudocode}
\usepackage{algorithm}

\usepackage{color}

\def\Tr{\operatorname{Tr}}
\def\Var{\operatorname{Var}}
\def\Cov{\operatorname{Cov}}
\def\uim{u_{im}}
\def\vim{v_{im}}
\def\uil{u_{il}}
\def\vil{v_{il}}
\def\cuim{\bar{u}_{im}}
\def\cvim{\bar{v}_{im}}
\def\cuil{\bar{u}_{il}}
\def\cvil{\bar{v}_{il}}
\def\ujm{u_{jm}}
\def\vjm{v_{jm}}

\def\vjl{v_{jl}}

\def\cvjm{\bar{v}_{jm}}
\def\cujl{\bar{u}_{jl}}

\title{Deflation as a Method of Variance Reduction for Estimating the Trace of a Matrix Inverse}
\author{Arjun Singh Gambhir \footnotemark[2]\ \footnotemark[3]
\and Andreas Stathopoulos \footnotemark[4]
\and Kostas Orginos \footnotemark[2]\ \footnotemark[3]}

\begin{document}
\maketitle
\renewcommand{\thefootnote}{\fnsymbol{footnote}}
\footnotetext[2]{Department of Physics, College of William and Mary, Williamsburg, Virginia 23187-8795, U.S.A.}
\footnotetext[3]{Jefferson National Laboratory, 12000 Jefferson Avenue, Newport News, Virginia, 23606, U.S.A.}
\footnotetext[4]{Department of Computer Science, College of William and Mary, Williamsburg, Virginia 23187-8795, USA}
\renewcommand{\thefootnote}{\arabic{footnote}}
\maketitle
\begin{abstract}
Many fields require computing the trace of the inverse of a large, sparse 
  matrix. 
Since dense matrix methods are not practical, the typical method used 
  for such computations is the Hutchinson method which is a Monte Carlo (MC)
  averaging over matrix quadratures.
To improve its slow convergence, several variance reductions techniques have 
  been proposed.
In this paper, we study the effects of deflating the near null 
  singular value space.
We make two main contributions: 
One theoretical and one by engineering a solution to a real world application.

First, we analyze the variance of the Hutchinson method as a function of 
  the deflated singular values and vectors. 
By assuming additionally that the singular vectors are random unitary matrices, 
  we arrive at concise formulas for the deflated variance
  that include only the variance and the mean of the singular values. 
We make the remarkable observation that deflation may increase variance 
  for Hermitian matrices but not for non-Hermitian ones. 
The theory can be used as a model for predicting the benefits of deflation.
Experimentation shows that the model is robust even when the singular vectors
  are not random.

Second, we use deflation in the context of a large scale application 
  of ``disconnected diagrams" in Lattice QCD.
On lattices, Hierarchical Probing (HP) has previously provided significant
  variance reduction over MC by removing ``error'' from neighboring nodes 
  of increasing distance in the lattice.
Although deflation used directly on MC yields a limited improvement of 30\% 
  in our problem, when combined with HP they reduce variance by a factor 
  of about 60 over MC. 
We explain this synergy theoretically and provide a thorough experimental
  analysis. 
One of the important steps of our solution is the pre-computation of 
  1000 smallest singular values of an ill-conditioned matrix of size 25 million.
Using the state-of-the-art packages PRIMME and a domain-specific 
  Algebraic Multigrid preconditioner, we solve this large eigenvalue 
  computation on 32 nodes of Cray Edison in about 1.5 hours and at a 
  fraction of the cost of our trace computation.

\end{abstract}

\section{Introduction}
The estimation of the trace of the inverse, $\Tr(A^{-1})$, 
  of a large, sparse matrix appears in many applications, 
  including statistics \cite{Hutchinson_90}, 
  data mining \cite{Bekas_diagonal}, and
  uncertainty quantification \cite{Bekas_Uncertain}.
Our application comes from Lattice quantum Chromodynamics (QCD),
in the context of which  {\it ab initio}   calculations of  properties of hadrons can be performed. 
This is fundamental for understanding the most basic properties 
  of matter \cite{Gupta:1997nd}. 
The size of the matrix, $N$, in these problems makes it prohibitive
  to use matrix factorization methods to compute the exact trace, and 
  usually the trace is not needed in high accuracy. 
For these reasons, the typical tool for such calculations has been a 
  Monte Carlo (MC) method due to Hutchinson \cite{Hutchinson_90}. 

The Hutchinson method estimates the trace of $A$ as $E( z^HAz)$, 
  i.e., the expectation value of quadratures of $A$ with random vectors. 
Thus, given a sequence of $s$ random vectors $z_j$ whose  
components are random variables that satisfy the property
$ E(z_j(k)z_j(k')) = \delta_{kk'}$, 
the following is an unbiased estimator for 
 $\Tr(A^{-1})$,
\begin{equation}
t(A^{-1}) = \frac{1}{s}\sum_{j=1}^{s}z_j^HA^{-1}z_j, \mbox{ with } E(t(A^{-1})) = \Tr(A^{-1}).
\end{equation}
Here a superscript $H$ denotes the Hermitian conjugate as we allow $A$ to be 
  complex. 
As in all MC processes, the error of the estimator reduces as 
  $\sqrt{\Var(t(A^{-1}))/s}$ which can be very slow.
Therefore, most effort has concentrated in reducing the variance of 
  the estimator \cite{Iitaka_trace, traceHongKong, Toledo_trace}.
\textcolor{black}{
For general matrices, minimum variance is achieved by 
  the $\mathbb{Z}_4$ noise vectors, whose elements are uniformly sampled 
  from $\{ \pm 1, \pm i\}$ \cite{Wilcox:1999ab}.
For real symmetric matrices and real random vectors the minimum is achieved
  by the Rademacher vectors \cite{Toledo_trace}, i.e.,  
  $\mathbb{Z}_2$ vectors with $\pm 1$ uniformly distributed elements.
However, the variance in this case is twice as large as the variance with
  $\mathbb{Z}_4$ vectors.
To treat the general case, we limit our discussion to $\mathbb{Z}_4$ vectors 
  for which the variance of the Hutchinson method is given by the following formula.
First, we define for any $A$ its traceless matrix 
  $\tilde A = A - \diag(\diag(A))$ using the MATLAB $\diag()$ operator.
Then,
\begin{equation}
\Var(t(A^{-1}))=\|\tilde{A^{-1}}\|_F^2=\|A^{-1}\|_F^2-\sum_{i=1}^{N}|A^{-1}_{i,i}|^2.
\label{eq:variance}
\end{equation}
Note that the variance for vectors with elements following a Gaussian distribution 
  is $2\|A^{-1}\|_F^2$ which can be much larger 
  than (\ref{eq:variance}) and also larger than $2\|\tilde{A^{-1}}\|_F^2$ in the more 
  common Hermitian, $\mathbb{Z}_2$ case.
On the other hand, reducing (\ref{eq:variance}) is more complicated
  than reducing $2\|A^{-1}\|_F^2$.
}

 From (\ref{eq:variance}), it is clear that large off diagonal elements of 
  $A^{-1}$ slow down the convergence of the MC estimator.
Therefore, it would be beneficial if the largest elements of $A^{-1}$
  were removed with some deterministic process, and MC were allowed to 
  converge on the remaining matrix. 
Several such variance reduction mechanisms have been discussed in the past. 
One idea is to use Hadamard (instead of random) vectors which annihilate 
  the contribution of specific diagonals of $A^{-1}$ \cite{Bekas_diagonal}.
This works only if the large magnitude elements of $A^{-1}$ happen to be
  on those diagonals.

An extension of that idea is probing \cite{Tang_Saad_traceInv}. 
If the sparsity pattern of $A$ is known, graph coloring can be used 
  to generate a number of vectors equal to the number of colors that, 
  if multiplied by $A$, reveal the exact diagonal of $A$.
However, $A^{-1}$ is dense in general so it is only meaningful to
  apply probing to a sparsified version of $A^{-1}$; one that keeps 
  the largest elements of $A^{-1}$.
A large class of matrices (including the Dirac operator in Lattice QCD) exhibit an exponential decay of the elements 
  $|A_{i,j}^{-1}|$ as a function of the graph theoretical distance dist$(i,j)$
  in the graph of $A$.
For such matrices, the non-zero structure of $A^k, k=1,2,\ldots$, 
  provides a natural sparsification pattern for $A^{-1}$.
Although the actual elements of $A^{-1}$ need not be known, coloring 
  the graph of $A^k$ provides the appropriate probing vectors that 
  annihilate the elements of $A^{-1}$ that reside on the structure of $A^k$.
This coloring is equivalent with a distance $k$ coloring of $A$.
The larger the distance, the more elements of $A^{-1}$ are captured, and 
  therefore removed from the variance.
This was followed in \cite{Tang_Saad_traceInv} but can be expensive 
  for any but the shortest distances.

A similar scheme called dilution has been used in Lattice QCD
  \cite{Babich:2011, Morningstar_Peardon_etal_2011}.
There, the Dirac matrix is a discrete differential operator defined
  on a regular 4D lattice, where each lattice site has 12 degrees of 
  freedom (representing a dimensionality of 4 for the ``spin'' space
  and 3 for the ``color'' space).
Therefore, partitioning the lattice in a red-black checkerboard
  can be used to remove the contribution to the variance from nearest 
  neighbor connections in $A^{-1}$.
Moreover, practitioners often ``dilute'' the spin and color components as well.
Clearly, this is equivalent to probing for distance 1 and has provided
  a good reduction of variance. In \cite{Babich:2011} more elaborate dilution patterns have also been explored. 

Hierarchical probing (HP) extends the idea of probing and dilution to all 
  distances $k=2^i$ up to the diameter of the lattice and forces the 
  colorings to be nested \cite{Stathopoulos:2013aci}.
This nesting allows us to reuse the linear system solutions with
  probing vectors from shorter distances, if higher accuracy is required 
  in the Hutchinson method.
These solutions cannot be reused with classical probing.
In addition, the HP computational cost for coloring lattices 
  for all possible distances is negligible and so is the cost for generating 
  the probing vectors using appropriate permutations of Hadamard or Fourier 
  matrices.
HP removes error contributions from the largest elements of $A^{-1}$ 
  incrementally and thus, for not too ill-conditioned matrices, 
  achieves an error that scales as ${\cal O}(1/s)$ instead of 
  ${\cal O}(1/\sqrt{s})$ in MC.
In Lattice QCD calculations, we have observed close to an order of 
  magnitude variance reduction with HP 
  \cite{Stathopoulos:2013aci,Green:2015wqa}.

 
In this paper we study the effect of deflation using a partial singular value 
  decomposition (SVD) of $A$ in reducing the variance of the Hutchinson method.
For ill conditioned matrices, the $A^{-1}$ is dominated by the near null 
  singular space of $A$, and thus does not display the decaying properties
  of its elements on which probing is based.
We expect that by deflating this dominant near null space the variance
  is reduced, and the nearest neighbor connections can again be exploited
  by probing.
For linear systems, deflation of the smallest singular triplets of $A$ improves 
  the condition number of $A$ and thus speeds up iterative methods.
Similarly, if Hutchinson is used with Gaussian random vectors, 
  because of the optimality properties of SVD, deflation reduces the 
  variance $2\|A^{-1}\|_F^2$.
Benefits were also observed when deflating other types of estimators 
  \cite{Tr_interpolate_jcp16}.
However, the situation for (\ref{eq:variance}) is more complicated.

Our first contribution is an analysis of the effects of SVD deflation 
  to the variance of the Hutchinson method for general matrices. 
We find that variance is reduced only if the singular values grow 
  at a \textcolor{black}{sufficiently} fast rate from small to large. 
To simplify the rather complicated formulas, we then assume that the singular 
  vectors are random unitary matrices, which is approximately true for 
  many large scale matrices and especially for our Lattice QCD application.
This leads to concise formulas for the deflated variance and to the
  remarkable observation that deflation may increase variance 
  for Hermitian matrices but not for non-Hermitian ones. 
The formulas can be used as a model for predicting and quantifying 
  the benefits of deflation.
Experimentation shows that the model is robust even when the singular vectors
  are not random.

Our second contribution is the combination of deflation and HP
  in Lattice QCD  calculations, where the computation of the trace of the inverse of the Dirac operator is required.
The resulting matrix is large (typical problems have matrices of dimension ${\cal O} (10^7)$) and ill-conditioned,
  which creates many problems for both HP and deflation.
The combination of the two methods has a much bigger effect than each method 
  alone, reducing variance in our test problem by a factor of over 150 compared to MC. 
We explain this synergy theoretically and provide a thorough experimental
  analysis. 
We also describe how we integrated and tuned two state-of-the-art software 
packages to solve our problem efficiently on the Cray \textcolor{black}{supercomputer} at NERSC (Edison).
The first is the PRIMME eigenvalue library \cite{PRIMME} 
  and the second is an Algebraic Multigrid preconditioner developed for 
  Lattice QCD \cite{Babich:2010qb}, which provides an average speedup 
  of 30 over unpreconditioned eigensolvers.
With these new methods and software we can now address efficiently 
  extremely challenging lattice problems.


\section{The effect of deflation on variance}
\label{sec:theory}
Deflation is typically understood as the removal of a certain eigenvector 
  space from the range of an operator \cite{Saad:2003:IMS:829576}.
The part of the spectrum removed is problem dependent.
For example, to speed up iterative methods for linear systems we 
  remove the smallest magnitude eigenvalues, while for variance reduction
  for $\Tr(A)$ we remove the largest magnitude.
In this paper we focus on deflation based on the space from the 
  singular value decomposition of $A$.

Let $A$ be a non-Hermitian matrix and assume without loss of generality 
  that we seek the $\Tr(A)$ ($A$ could be the inverse or some other 
  function of our matrix).
Let $A = U\Sigma V^T$ be the singular value decomposition (SVD) of $A$, and
  $U_1, V_1, \Sigma_1$ the $k$ largest singular triplets.
If $U=[U_1,U_2], V = [V_1,V_2]$, and $\Sigma = \mbox{diag}(\Sigma_1,\Sigma_2)$, 
  $A$ can be decomposed as
\begin{equation}
A=U_1\Sigma_1V_1^H+U_2\Sigma_2V_2^H \equiv A_D+A_R,
\label{eq:deflatedA}
\end{equation}
and $\Tr(A)=\Tr(A_D)+\Tr(A_R)$. 
If the triplet $U_1, V_1, \Sigma_1$ has been pre-computed, using the 
  cyclic property of the trace, we can explicitly compute 
  $\Tr(A_D) =\Tr(\Sigma_1 V_1^HU_1)$ with only $O(k^2N)$ operations.
What remains is to compute $\Tr(A_R)$, the trace of our matrix projected 
  on the remaining singular vectors, $A_R =A(I-V_1V_1^H)$.
This can be estimated stochastically.
Because $A_D$ is the best $k$-rank approximation of $A$ in the least 
  squares sense, one would expect that the variance on $A_R$ will always 
  be smaller than (\ref{eq:variance}).
We show next that this only happens under certain conditions. 

\begin{theorem}
Let ($\sigma_1 \geq \ldots \geq \sigma_N \geq 0$, $U$, $V$) be the singular 
  triplets of $A$ and $\Delta = (U\odot \bar{V})^H(U\odot \bar{V})$, where
  $\odot$ is the elementwise product of matrices, and
  $\bar V$ is the elementwise conjugate of $V$.
This gives,
\begin{equation}
\Delta_{ml} = \sum_{i=1}^N \cuim \vim \uil \cvil, \qquad  m,l = 1,\ldots ,N.
	\label{eq:Delta}
\end{equation}
Consider the decomposition in (\ref{eq:deflatedA}) produced by deflating the 
  largest $k$ singular triplets.
The variance of the stochastic estimator for $\Tr(A_R)$ satisfies,
\begin{equation}
	\Var(t(A_R)) = \sum_{m=k+1}^N \sigma_m^2 - 
	\sum_{m=k+1}^N \sum_{l=k+1}^N \sigma_m \sigma_l \Delta_{ml}. \label{eq:varAr}
\end{equation}
The variance of the stochastic estimator for $\Tr(A)$ follows from (\ref{eq:varAr}) with $k=0$.
Then, their difference is
\begin{eqnarray}
\label{eq:vA-vAR}
\Var(t(A)) - \Var(t(A_R))
   &=& \sum_{m=1}^k \sigma_m^2 (1- \Delta_{mm}) 
        - \sum_{m=1}^k 
	\sum_{l=m+1}^{N} \sigma_m\sigma_l(\Delta_{ml}+\Delta_{lm}).
\end{eqnarray}
\label{theorem:normDiff}
\end{theorem}
\begin{proof}
\textcolor{black}{Define} the vectors
$D=\diag(A)$ 
and $D_R=\diag(A_R)$ and the traceless 
matrices
$\tilde A = A - \diag(D)$
and $\tilde{A_R} = {A_R} - \diag(D_R)$.
According to (\ref{eq:variance}), we need to obtain an expression for 
$\Var(t(A)) - \Var(t(A_R)) = \|\tilde{A}\|^2_F - \|\tilde{A_R}\|^2_F$.
 From the properties of the SVD we have
\begin{equation}
	\|A\|_F^2 = \sum_{i=1}^N \sigma_i^2 
		= \|\tilde{A}\|^2_F + \|D\|^2_F, \qquad
	\|A_R\|_F^2 = \sum_{i=k+1}^N \sigma_i^2
                    = \|\tilde{A_R}\|^2_F + \|D_R\|^2_F.
	\label{eq:normA-normAR}
\end{equation}
Next we represent the diagonals in terms of the SVD
\begin{equation}
D(i)    =\sum_{m=1}^N\sigma_m \uim \cvim, \qquad
D_{R}(i)=\sum_{m=k+1}^N\sigma_m \uim \cvim .
\end{equation}
Then, using (\ref{eq:Delta}) we obtain expressions for the norms of the diagonals,
\begin{eqnarray}
\begin{array}{rcl}
\|D\|^2_F &=&  \sum_{i=1}^N  
        (\sum_{m=1}^N \sigma_m \cuim \vim)
        (\sum_{l=1}^N \sigma_l \uil \cvil)
    = 
	\sum_{m=1}^N \sum_{l=1}^N \sigma_m \sigma_l  
		\sum_{i=1}^N   \cuim \vim \uil \cvil 	\\
   & = & 
	\sum_{m=1}^N \sum_{l=1}^N \sigma_m \sigma_l \Delta_{ml}.
\end{array}
\end{eqnarray}
\begin{eqnarray}
\begin{array}{rcl}
\|D_R\|^2_F & =&  \sum_{i=1}^N  
        (\sum_{m=k+1}^N \sigma_m \cuim \vim)
        (\sum_{l=k+1}^N \sigma_l \uil \cvil)
    = \sum_{m=k+1}^N \sum_{l=k+1}^N \sigma_m \sigma_l \Delta_{ml}.\\
\end{array}
\end{eqnarray}
Utilizing the above and according to (\ref{eq:variance}) and (\ref{eq:normA-normAR}),
  the variance of the trace estimator for the deflated problem is given by (\ref{eq:varAr}),
  and similarly for the original problem with $k=0$.
To obtain the difference of these variances we denote 
  $s_{ml} = \sigma_m \sigma_l \Delta_{ml}$ for brevity and perform the required algebraic operations,
\begin{eqnarray}
\begin{array}{rcl}
\Var(t(A)) - \Var(t(A_R)) 
	&=&  \sum_{m=1}^k \sigma_m^2 
	-\sum_{m=1}^k \sum_{l=1}^N s_{ml}
	-\sum_{m=k+1}^N ( \sum_{l=1}^N s_{ml} -\sum_{l=k+1}^N s_{ml})\\
	&=&  \sum_{m=1}^k \sigma_m^2 
	-\sum_{m=1}^k \sum_{l=1}^N s_{ml} -\sum_{m=k+1}^N \sum_{l=1}^k s_{ml}\\
	&=&  \sum_{m=1}^k \sigma_m^2 
	-\sum_{m=1}^k (\sum_{l=1}^N s_{ml} +\sum_{l=k+1}^N s_{lm}) \\
	&=&  \sum_{m=1}^k \sigma_m^2 
	-\sum_{m=1}^k (\sum_{l=1}^k s_{ml} + \sum_{l=k+1}^N (s_{ml}+s_{lm})) \\
	&=&  \sum_{m=1}^k \sigma_m^2 -\sum_{m=1}^k s_{mm} 
	-\sum_{m=1}^k (\sum_{l=1,l\neq m}^k s_{ml} + \sum_{l=k+1}^N (s_{ml}+s_{lm})) \\
	&=&  \sum_{m=1}^k \sigma_m^2 (1-\Delta_{mm})
	-\sum_{m=1}^k \sum_{l=m+1}^N (s_{ml}+s_{lm}),
\end{array} \nonumber
\end{eqnarray}
which yields the desired result.
\end{proof}

{\bf Example.}
To achieve variance reduction, we need 
   $\|\tilde{A}\|_F^2-\|\tilde{A_R}\|_F^2>0$.
Contrary to low rank matrix approximations,
  deflation may not achieve this for $\Var(t(A_R))$.
Consider the following example in MATLAB: 
\begin{verbatim}
     [U,~] = qr([ -1 d d; d 1 1; d 1 -1]);
     A  = U*diag([1+2*s, 1+s, 1])*U';
     Ar = U(:,2:3)*diag([1+s, 1])*U(:,2:3)';
     disp(norm(A-diag(diag(A)),'fro')^2/norm(Ar-diag(diag(Ar)),'fro')^2);
\end{verbatim}
With $d$, we control the distance of the {\tt U(:,1)} and {\tt U(:,2:3)} 
  singular subspaces (also eigenspaces) from the first orthocanonical vector.
With $s$, we control the separation of the singular values.
We deflate the largest singular triplets.
For $s=0.5$ (i.e., singular values $[2, 1.5, 1]$), 
  we can verify numerically that $\Var(t(Ar))\geq \Var(t(A))$ for any $d$.
Deflation has a negative effect!
Similarly, if $d = 0.001$ there is no reduction of variance 
  regardless of the separation of the singular values.
On the other hand, for $s=1$  (i.e., singular values $[3, 2, 1]$), 
  we have $\Var(t(Ar))\geq \Var(t(A))$ for $d\leq1$ and 
  $\Var(t(Ar))<\Var(t(A))$ for $d>1$.
Finally, for $s=2$ (i.e., singular values $[5, 3, 1]$) 
  deflating the largest triplet reduces variance for all $d$.

Theorem \ref{theorem:normDiff} differs in two ways from the typical 
SVD based low rank approximations. 
The first is the $(1-\Delta_{mm})$ factor on the sum of $\sigma_m^2$ in (\ref{eq:vA-vAR}).
If $\Delta_{mm}\approx 1$ then $A$ is almost decomposable and therefore
  deflation will not remove any off-diagonal elements from the variance.
However, in this uncommon case, the deflated triplet did not 
  contribute to the variance in the beginning.
The second difference is the subtraction of the double summation term.
The hope is that the deflated singular values are sufficiently large to 
  dominate over the summation of $\sigma_m\sigma_l$.
However, this is complicated by the presence of the $\Delta_{ml}$. 

We attempt to analyze the condition for variance reduction further.
If we rewrite (\ref{eq:vA-vAR}) as,
$
\Var(t(A)) - \Var(t(A_R)) = 2\sum_{m=1}^k \sigma_m
 	( \sigma_m (1- \Delta_{mm}) -
	\sum_{l=m+1}^{N} \sigma_l(\Delta_{ml}+\Delta_{lm}))
$,
we observe that for the difference to be positive, a sufficient but not 
  necessary condition is for every term in the sum to be positive. 
Equivalently, 
\begin{equation}
 	\sigma_m > \sum_{l=m+1}^{N} \sigma_l
		\frac{(\Delta_{ml}+\Delta_{lm})}{(1- \Delta_{mm})},
	\qquad m=1,\ldots,k.
	\label{eq:crudeModel}
\end{equation}
To simplify further, assume that $U=V$, i.e., the matrix is Hermitian.
Then, $\Delta_{ml} = \sum_{i=1}^N |\uim|^2|\uil|^2 \geq 0$, and therefore
$\sum_{m=1}^N \Delta_{ml} = \sum_{l=1}^N \Delta_{ml}=1$, or
 $\Delta$ is doubly stochastic.
For $k=1$, we have
$\sum_{l=2}^N \Delta_{1l} = 1-\Delta_{11}$ and thus if 
$\sigma_1 > 2 \sigma_2$, then
$
\sigma_1 > 2 \sigma_2 \frac{1-\Delta_{11}}{1-\Delta_{11}} 
         = 2 \sigma_2 \sum_{l=2}^{N} 
		\frac{\Delta_{1l}}{(1- \Delta_{11})}
	 > \sum_{l=2}^{N} \sigma_l
		\frac{(\Delta_{1l}+\Delta_{l1})}{(1- \Delta_{11})}
$.
So if $\sigma_1 > 2 \sigma_2$ then deflating the largest singular
  triplet reduces the variance.
Inductively, if the singular spectrum decreases geometrically as 
  $\sigma_{i+1} = 2^{-i}\sigma_1$, we 
  guarantee that any deflation improves variance. 
This requirement, however, is too pessimistic.
In the following we show that if $U$ and $V$ are random unitary matrices, 
  which is approximately true in our LQCD application,
  the $\Delta_{ml}$ are uniformly small with small variance.

\subsection{The case of random singular vectors}
\label{sec:randomtheory}
Let us assume that $U$ and $V$ are standard unitary matrices, i.e., 
  distributed with the Haar probability measure \cite{Hiai-Petz-book}.
In LQCD, the $U$ and $V$ are indeed random matrices, although they do not 
  follow exactly the above distribution.
Our resulting model, however, will not depend on the nature of 
  the distribution but on the expectation and variance of the elements of
  $U$ and $V$, based on which we will provide a bound for the elements 
  of $\Delta$ as a random variable.
Specifically, we base our subsequent analysis on 
  $|\Delta_{ml} - E(\Delta_{ml})| = O(\sqrt{\Var(\Delta_{ml}))}$.
First we need the following.

\begin{proposition}\mbox{\rm \cite[Prop. 4.2.2]{Hiai-Petz-book}}.
Let 
$l\in \mathbb{N}, i_1,\ldots ,i_l,j_1,\ldots,j_l \in \left\{1,\ldots,N\right\}$
and $k_1,\ldots ,k_l,m_1,\ldots,m_l$ $\in \mathbb{Z}^+$. If either 
$\sum_{i_r=i}(k_r-m_r) \neq 0$ for some $1\leq i \leq N$ or 
$\sum_{j_r=i}(k_r-m_r) \neq 0$ for some $1\leq j \leq N$, then
$$E((u^{k_1}_{i_1j_1}\bar{u}^{m_1}_{i_1j_1})
    (u^{k_2}_{i_2j_2}\bar{u}^{m_2}_{i_2j_2}) \cdots
    (u^{k_l}_{i_lj_l}\bar{u}^{m_l}_{i_lj_l})) = 0.$$
\label{prop:book}
\end{proposition}
Then we have the following.
\begin{lemma}
For an $N\times N$ standard unitary matrix $U = [u_{ij}]$ it holds:
\begin{eqnarray}
E(u_{ij}) & = & 0,  \label{eq:meanUij}\\
E(u_{ij}\bar{u}_{kj}) & = & 0  \quad (i\neq k), \label{eq:meanUijUkj}\\
E(u_{ij}\bar{u}_{kj}u_{im}\bar{u}_{km}) & = & 0 \quad (i\neq k),\label{eq:meanUijUkjUimUkm}\\
E(|u_{ij}|^2) & = & \frac{1}{N}, \label{eq:varUij}\\
E(|u_{ij}|^4) & = & \frac{2}{N(N+1)}, \label{eq:4momentUij}\\
E(|u_{ij}|^8) & = & \frac{4!}{N(N+1)(N+2)(N+3)}, \label{eq:8momentUij}\\
E(|u_{ij}|^2|u_{mj}|^2) & = & E(|u_{ij}|^2|u_{ik}|^2)  = \frac{1}{N(N+1)}
\quad (i\neq m, j\neq k), \label{eq:meanUij2Umj2}\\
E(|u_{ij}|^4|u_{mj}|^4) & = & E(|u_{ij}|^4|u_{ik}|^4) = E(|u_{ij}|^8)/6
\quad (i\neq m, j\neq k). \label{eq:meanUij4Umj4}
\end{eqnarray}
\end{lemma}
\begin{proof}
Using Proposition \ref{prop:book} with $l=1, k_1 = 1, m_1 = 0, i_1 = i$,
  and $j_1=j$ gives (\ref{eq:meanUij}).

Using the same proposition with 
  $l=2, i_1=i, j_1=j_2=j, i_2 = k, m_1=0, k_1=1, m_2=1, k_2=0$
  and picking $i_r = i_1=i$, gives (\ref{eq:meanUijUkj}).

Choosing $l=4, i_1=i_3=i, j_1=j_2=j, i_2=i_4=k, j_3=j_4=m$, $k_1 = k_3 = 1, k_2=k_4=0$
and $m_1=m_3=0, m_2=m_4=1$ and picking $i_r = i$ with $r \in \left\{1,3\right\}$
gives (\ref{eq:meanUijUkjUimUkm}).

Equations (\ref{eq:varUij}), (\ref{eq:4momentUij}), (\ref{eq:meanUij2Umj2}) 
  are borrowed from \cite[Prop.4.2.3, p.138]{Hiai-Petz-book}.
Lemma 4.2.4 in \cite{Hiai-Petz-book} states:
\begin{equation}
E(|U_{ij}|^{2k}) = \left(
\begin{array}{c}
N+k-1\\
N-1
\end{array}
\right)^{-1}.
\label{eq:moment2k}
\end{equation}
By setting $k=4$ we obtain (\ref{eq:8momentUij}). 
The derivation of (\ref{eq:meanUij4Umj4}) is based on the proof of 
  Proposition 4.2.3 in the above book but is more tedious.
We exploit the idea that $\uim$ and $(\uim \cos\theta + \ujm \sin\theta)$
  are identically distributed, and consequentially have the same expectations and 
  moments.
After algebraically expanding the eighth moment of the absolute value, 
  we observe that most terms vanish because of Proposition \ref{prop:book}. 
The surviving terms are the following:
\begin{eqnarray}
E(|\uim|^8) &=& E(|\uim \cos\theta + \ujm \sin\theta|^8) \nonumber \\
	    &=& E(|\uim|^8\cos^8\theta) + E(|\ujm|^8 \sin^8\theta) + 16E(|\uim|^2|\ujm|^6\cos^2\theta \sin^6\theta) \nonumber \\
	    & & + 16E(|\uim|^6|\ujm|^2\cos^6\theta  \sin^2\theta) + 36E(|\uim|^4|\ujm|^4\cos^4\theta  \sin^4\theta).
	\label{eq:umi8-1}
\end{eqnarray}
First, we note that  $E(|\uim|^8) = E(|\ujm|^8)$ and similarly $E(|\uim|^2|\ujm|^6) = E(|\uim|^6|\ujm|^2)$.
Then by simple integration on $[0,2\pi]$ we obtain,
$E(\cos^8\theta+\sin^8\theta) = 35/64, \
E(\cos^2\theta\sin^6\theta) = E(\cos^6\theta\sin^2\theta) = 5/128$, and 
$E(\cos^4\theta\sin^4\theta) = 3/128$.
By substituting in (\ref{eq:umi8-1}) we obtain,
\begin{equation}
29 E(|\uim|^8) = 80 E(|\uim|^2|\ujm|^6) + 54 E(|\uim|^4|\ujm|^4).  \label{eq:umi8-2}
\end{equation}
We still need to determine the term  $E(|\uim|^2|\ujm|^6)$.
We follow the same procedure, noting also that $\ujm$ is identically distributed
with $-\sin\theta\uim+\cos\theta\ujm$. By expanding the following moments and canceling 
  several terms because of Proposition \ref{prop:book} we have,
\begin{eqnarray}
E(|\uim|^2|\ujm|^6) 
           &=& E(|\uim\cos\theta+\ujm\sin\theta|^2|\ujm\cos\theta-\uim\sin\theta|^6) \nonumber \\
           &=& E(\uim^8\cos^2\theta\sin^6\theta) + E(\ujm^8\cos^6\theta\sin^2\theta) 
		          + E(\uim^2\ujm^6\cos^8\theta) + E(\uim^6\ujm^2\sin^8\theta \nonumber)\\
	   & & +E(\uim^2\ujm^6(9\cos^4\theta\sin^4\theta - 6\cos^6\theta\sin^2\theta))
			+E(\uim^6\ujm^2(9\cos^4\theta\sin^4\theta-6\cos^2\theta\sin^6\theta)) \nonumber\\
	   & & +E(\uim^4\ujm^4(9\cos^2\theta\sin^6\theta- 18\cos^4\theta\sin^4\theta +
							9\cos^6\theta\sin^2\theta)). \label{eq:umi2ujm6-1}
\end{eqnarray}
As before, we consolidate all expectations that are the same and integrate to find the following expectations:
$E(\cos^2\theta\sin^6\theta)+\cos^6\theta\sin^2\theta)= 5/64,~
E( \cos^8\theta + \sin^8\theta + 18\cos^4\theta\sin^4\theta - 6\cos^6\theta\sin^2\theta - 6\cos^2\theta\sin^6\theta)= 1/2,$ and
$E(9\cos^2\theta\sin^6\theta - 18\cos^4\theta\sin^4\theta + 9\cos^6\theta\sin^2\theta) = 9/32$. 
This yields the following equation:
\begin{equation}
E(|\uim|^2|\ujm|^6)  = 5/32 E(|\uim|^8) + 9/16 E(|\uim|^4|\ujm|^4).  \label{eq:umi2umk6-2}
\end{equation}
By substituting (\ref{eq:umi2umk6-2}) into (\ref{eq:umi8-2}) we obtain the desired (\ref{eq:meanUij4Umj4}).
\end{proof}

{\bf Remark 1.} When $U$ is a real orthogonal matrix the formulas 
(\ref{eq:meanUij})--(\ref{eq:varUij}) and (\ref{eq:meanUij2Umj2}) hold, but the fourth 
moment is given by $E(u_{ij}^4) = 3/(N^2+N)$. 
In the proof of the Proposition 4.2.3 in \cite[p.139]{Hiai-Petz-book} 
  the term $E(u_{11}^2\bar{u}_{21}^2) + E(\bar{u}_{11}^2u_{21}^2)$ 
  vanishes in the unitary case, but in the real case it
  becomes $2E(u_{11}^2u_{21}^2)=2/(N^2+N)$.

{\bf Remark 2.}
The results of the above Lemma can also be obtained by physics and 
combinatorial considerations as for example in Lattice QCD \cite{creutz1983quarks}.

To study $\Delta_{ml}$, we first assume that for general non-Hermitian matrices
  $U$ is statistically independent from $V$. 
We address Hermitian matrices ($U=V$) separately.
Statistical independence cannot be claimed between all the elements of $U$.
In fact, we know that the elements of up to $k\times k$ submatrices of $U$, 
where $k=o(\sqrt{N})$, can be approximated by independent Gaussians, 
${\cal N}(0,1/\sqrt{N})$ \cite{2006math......1457J}. 
We will not use this result because our formulas involve more than 
  ${\cal O}(N)$ elements of $\Delta$. 
Instead, we will find the expected
  value and variance of elements $\Delta_{ml}$.

\begin{lemma} 
Let $\Delta_{ml}$ be defined in (\ref{eq:Delta}) and assume $m\neq l$. 
For non-Hermitian matrices it holds,
\begin{eqnarray}
	E(\Delta_{ml})    &=& 0,\\
	\Var(\Delta_{ml}) &=& 1/(N(N+1)^2),\\
	E(\Delta_{mm})    &=& 1/N,\\
	\Var(\Delta_{mm}) &=& (N-1)/(N^2(N+1)^2).
\end{eqnarray}
\label{lemma:DeltaNonHerm}
\end{lemma}
\begin{proof}
If $m\neq l$, (\ref{eq:meanUijUkj}) implies
  $E(\Delta_{ml}) =  \sum_{i=1}^N E(\cuim \uil) E(\vim \cvil) 
	= \sum_{i=1}^N \overline{E(\uim \cuil)} E(\vim \cvil) = 0$.
 From (\ref{eq:varUij}) we have
 $\Delta_{mm} = \sum_{i=1}^N E(\cuim \uim) E(\vim \cvim) = \sum_{i=1}^N 1/N^2 = 1/N$. \\
For the variance of non-diagonal elements we have,
\begin{equation*}
\begin{array}{lclr}
   \Var(\Delta_{ml}) &=& E(\Delta_{ml}\overline{\Delta_{ml}})-|E(\Delta_{ml})|^2 \\
   &=& E((\sum_{i=1}^N \cuim \uil \vim \cvil) (\sum_{i=1}^N \uim \cuil \cvim \vil)) 
	\qquad & (\mbox{since }E(\Delta_{ml}=0) \\
   &=& \sum_{i,j=1}^N E(\cuim \uil \ujm \cujl)E(\vim \cvil \cvjm \vjl)\\
   &=& \sum_{i=1}^N E(\cuim \uil \uim \cuil)E(\vim \cvil \cvim \vil)  
	\qquad & (i\neq j \mbox{ terms are 0 from (\ref{eq:meanUijUkjUimUkm}))} \\
	&=& \sum_{i=1}^N E(|\uil|^2 |\uim|^2)E(|\vim|^2 |\vil|^2) = \frac{1}{N(N+1)^2}
	\qquad & (\mbox{from (\ref{eq:meanUij2Umj2}))}. 
\end{array}
\end{equation*}
For the diagonal we use (\ref{eq:meanUij2Umj2}) and (\ref{eq:4momentUij}), 
\begin{equation*}
\begin{array}{lclr}
   \Var(\Delta_{mm}) &=& E(\Delta_{mm}\overline{\Delta_{mm}})-|E(\Delta_{mm})|^2 \\
     &=& E((\sum_{i=1}^N |\uim|^2 |\vim|^2) (\sum_{i=1}^N |\uim|^2 |\vim|^2))-1/N^2 \\
     &=& \sum_{i=1,j=1}^N E(|\uim|^2 |\ujm|^2)E(|\vim|^2 |\vjm|^2) -1/N^2 \\
     &=& \sum_{i=1,j=1,i\neq j}^N E(|\uim|^2 |\ujm|^2)E(|\vim|^2 |\vjm|^2) +
	\sum_{i=1}^N E(|u_{mm}|^4)E(|v_{mm}|^4) -1/N^2  \\
	&=& \frac{N^2-N}{N^2(N+1)^2} + \frac{4N}{ N^2(N+1)^2} -\frac{1}{N^2} = \frac{N-1}{N^2(N+1)^2}.
\end{array}
\end{equation*}
\end{proof}
\begin{lemma} Let $\Delta_{ml}$ be defined in (\ref{eq:Delta}) and assume $m\neq l$. 
For Hermitian matrices ($U=V$) it holds,
\begin{eqnarray}
	E(\Delta_{ml})    &=& 1/(N+1),\\
	\Var(\Delta_{ml}) &=& 2(N-1)/((N+1)^2(N+2)(N+3)),\\
	E(\Delta_{mm})    &=& 2/(N+1),\\
	\Var(\Delta_{mm}) &=&  4(N-1)/((N+1)^2(N+2)(N+3)).
\end{eqnarray}
\label{lemma:DeltaHerm}
\end{lemma}
\begin{proof}
When $U=V$, $\Delta_{ml} = \sum_{i=1}^N |\uim|^2|\uil|^2$.
Moreover, $\sum_{m=1}^N \Delta_{ml} = \sum_{l=1}^N \Delta_{ml}=1$, and because $\Delta_{ml} \geq 0$, 
 $\Delta$ is a doubly stochastic matrix. 
 From (\ref{eq:meanUij2Umj2}) we have 
  $E(\Delta_{ml}) = \sum_{i=1}^N E(|\uim|^2 |\uil|^2) = 1/(N+1)$.
 From (\ref{eq:4momentUij}) we have $E(\Delta_{mm}) = \sum_{i=1}^N E(|\uim|^4) = 2/(N+1)$.
 We turn our attention to the variance of $\Delta_{mm}$, which is larger than of $\Delta_{ml}$.
Using (\ref{eq:8momentUij}) and (\ref{eq:moment2k}) we have,
\begin{equation*}
\begin{array}{lclr}
\Var(\Delta_{mm}) &=&  E((\sum_{i=1}^N |\uim|^4)^2) - |E(\Delta_{mm})|^2
     		= \sum_{i=1,j=1}^N E(|\uim|^4 |\ujm|^4) -4/(N+1)^2 \\
     &=& \sum_{i=1,j=1,i\neq j}^N  E(|\uim|^4 |\ujm|^4) +  \sum_{i=1}^N E(|\uim|^8) - 4/(N+1)^2\\
     &=& (N^2-N)E(|\uim|^8)/6 +   N E(|\uim|^8) - 4/(N+1)^2\\
     &=& 4(N+5)/((N+1)(N+2)(N+3)) - 4/(N+1)^2 = 4(N-1)/((N+1)^2(N+2)(N+3)).
\end{array}
\end{equation*}
A similar but lengthier strategy yields also $\Var(\Delta_{ml})$.
\end{proof}

We are now able to characterize the effect of deflation on matrices 
  with random singular vectors.
We use the variance of the elements of $\Delta$ to ascertain their magnitude as
  $|\Delta_{ml} - E(\Delta_{ml})| = O(\sqrt{\Var(\Delta_{ml}))}$.
The above results show that in the non-Hermitian case 
$|\Delta_{ml}| = \Theta(1/N^{1.5})$, while in the Hermitian case 
$|\Delta_{ml} - 1/(N+1)| = \Theta(1/N^{1.5})$, 
\textcolor{black}{
   and thus $\Delta_{ml} = 1/(N+1) + \Theta(1/N^{1.5})$.
The diagonal elements in both cases are around $1/N$ and $1/(N+1)$ respectively.
}
Using these formulas for $\Delta_{ml}$ and for $\Delta_{mm}$, we revisit
  the sufficient but pessimistic condition of (\ref{eq:crudeModel}) for non-Hermitian and Hermitian
  matrices,
\begin{equation}
	\mbox{non-Hermitian:  }
        \sigma_m > \Theta(\frac{1}{N^{1.5}}) \sum_{l=m+1}^{N} \sigma_l,
	\quad
	\mbox{Hermitian:  }
	\textcolor{black}{
	\sigma_m > (\frac{1}{N-1}+\Theta(1/N^{1.5})) \sum_{l=m+1}^{N} \sigma_l,
	}
        \quad m=1,\ldots,k.
\label{eq:crudeInequalities}
\end{equation}
We are seeking the singular value distributions that would satisfy (\ref{eq:crudeInequalities}).
If we model the summations as $\int_{m+1}^N \sigma(x) dx$, we can readily verify 
  that the least decaying series that satisfy the inequalities for all $m$ are 
\begin{equation}
	\mbox{non-Hermitian:  }
	\sigma_{i} = \Theta(\sqrt{N-i+1}),\qquad
	\mbox{Hermitian:  }
	\sigma_{i} = \Theta(N-i+1).
	\label{eq:improvedPessimistic}
\end{equation}
It is remarkable that it is harder for Hermitian matrices to achieve variance reduction; 
  in other words the deflated singular values must decay much faster (have larger separations)
  to achieve the same variance reduction as in a non-Hermitian matrix.
On the other hand, the $\Delta_{mm}$ and $\Delta_{ml}$ are positive and larger
  for Hermitian than non-Hermitian matrices, which implies that 
  the subtracting term in (\ref{eq:vA-vAR}) is always larger for 
  Hermitian matrices.
Therefore, a Hermitian matrix is expected to have lower starting variance 
  than a non-Hermitian matrix with the same singular spectrum.
We conclude that although non-Hermitian matrices outperform Hermitian ones 
  in variance reduction, it is because they have more variance to reduce.

The above analysis is intuitively useful, but dependent on the pessimistic
  condition (\ref{eq:crudeModel}).
The following theorem gives the expected variance of our trace estimator as 
  an expression of only the mean and variance of the singular values.
Because of the small variance of the $\Delta_{ml}$ elements in 
  Lemma \ref{lemma:DeltaHerm}, the expected variance is very accurate.
\begin{theorem}
Define the mean and the variance of the $N-k$ singular values of $A_R$,
$\mu_k =\frac{1}{N-k} \sum_{m=k+1}^{N}\sigma_m$, and
$V_k = \frac{1}{N-k} \sum_{m=k+1}^N (\sigma_m - \mu_k)^2$, respectively.
Then, for non-Hermitian matrices it holds
\begin{equation*}
E(\Var(t(A_R))) = 
	(N-k) (1-\frac{1}{N}) (V_k + \mu_k^2)
\end{equation*}
and for Hermitian matrices,
\begin{equation*}
E(\Var(t(A_R))) = 
	(N-k)\left( V_k \frac{N}{N+1} + \mu_k^2 \frac{k}{N+1} \right).
\end{equation*}
In addition, the relative standard deviation of our variance estimator, $\Var(t(A_R))$, 
 is bounded by
\begin{equation*}
	\frac{\operatorname{StdDev}( \Var(t(A_R)) )} {E(\Var(t(A_R))) }
	\leq 
	{\cal O}(\frac{N-k}{N^{1.5}}).
\end{equation*}
\label{thm:exp-vars}
\end{theorem}
\begin{proof}
First note that 
$V_k = \frac{1}{N-k} \left( \sum_{m=k+1}^N \sigma_m^2 
	- \frac{1}{N-k}\sum_{m, l=k+1}^N \sigma_m \sigma_l \right), $
which gives
\begin{eqnarray}
\sum_{m=k+1}^N \sigma_m^2  &=&
   (N-k) V_k + 
   \frac{1}{N-k}\sum_{m=k+1}^N \sum_{l=k+1}^N \sigma_m \sigma_l 
   = (N-k) V_k + (N-K) \mu_k^2.
	\label{eq:SVvar}
\end{eqnarray}
Taking expectation values in (\ref{eq:varAr}) we have,
\begin{eqnarray} 
E(\Var(t(A_R)))
     &=& 
	\sum_{m=k+1}^N \sigma_m^2 
	- E(\Delta_{mm}) \sum_{m=k+1}^N \sigma_m^2  
	- E(\Delta_{ml}) \sum_{m=k+1}^N \sum_{l=k+1,l\neq m}^N \sigma_m \sigma_l.
\label{eq:intermediate-exp}
\end{eqnarray}
Then, for non-Hermitian matrices (\ref{eq:intermediate-exp}), (\ref{eq:SVvar}) 
and Lemma \ref{lemma:DeltaHerm} yield,
\begin{equation*} 
\begin{array}{c}
	E(\Var(t(A_R))) =
	\sum_{m=k+1}^N \sigma_m^2(1- \frac{1}{N}) = (N-k) (1-\frac{1}{N}) (V_k + \mu_k^2).
\end{array}
\end{equation*}
Similarly, for Hermitian matrices we have,
\begin{equation*}
\begin{array}{rcl}
E(\Var(t(A_R)))
     &=& 
	\sum_{m=k+1}^N \sigma_m^2 
	- \frac{2}{N+1} \sum_{m=k+1}^N \sigma_m^2  
	- \frac{1}{N+1} \sum_{m=k+1}^N \sum_{l=k+1,l\neq m}^N \sigma_m \sigma_l\\
     &=&  
	\sum_{m=k+1}^N \sigma_m^2  (1- \frac{1}{N+1})
	- \frac{1}{N+1} \sum_{m=k+1}^N \sum_{l=k+1}^N \sigma_m \sigma_l \\
     &=&  
	(N-k)V_k(1-\frac{1}{N+1}) + (N-k)\mu_k^2(1-\frac{1}{N+1})  
	- \frac{(N-k)^2}{N+1} \mu_k^2 \nonumber \\
     &=&  
	(N-k)V_k(1-\frac{1}{N+1}) + k\frac{N-k}{N+1} \mu_k^2.\nonumber
\end{array}
\end{equation*}
To gauge the accuracy of the above estimation, we need to compute the variance of 
  our variance approximation. 
First note that for both non-Hermitian and Hermitian matrices, $\Var(\Delta_{ml}) \leq c/N^3$, 
  with $c$ being the maximum of the variance constants in Lemmas \ref{lemma:DeltaNonHerm} 
  and \ref{lemma:DeltaHerm}.
Then, using the rule for the variance of the sum of random variables we get,
\begin{equation*}
\begin{array}{rcl}
\Var( \Var(t(A_R)) ) &=& 
	\Var( \sum_{m=k+1}^N \sum_{l=k+1}^N \sigma_m \sigma_l \Delta_{ml} ) 
	\; = \;
	\sum_{i,j=k+1}^N \sum_{m,l=k+1}^N 
	\sigma_i \sigma_j \sigma_m \sigma_l \Cov(\Delta_{ij},\Delta_{ml}) \\
	&\leq&
	\sum_{i,j=k+1}^N \sum_{m,l=k+1}^N 
	\sigma_i \sigma_j \sigma_m \sigma_l \sqrt{\Var(\Delta_{ij})\Var(\Delta_{ml})} 
	\; \leq\;
	\frac{c}{N^3} (\sum_{m,l=k+1}^N \sigma_m \sigma_l)^2 \\
	&=&
	\frac{c(N-k)^4}{N^3} \mu_k^4.
\end{array}
\end{equation*}
Since $E(\Var(t(A_R))) \geq (N-k)(V_k+\mu_k^2) \geq (N-k)\mu_k^2$ for both non-Hermitian and 
Hermitian matrices, the relative error of using $E(\Var(t(A_R)))$ instead of 
$\Var(t(A_R))$ can be bounded as,
\begin{equation*}
	\frac{\sqrt{\Var( \Var(t(A_R)) )}}
	{E(\Var(t(A_R))) }
	\leq
	\frac{ \frac{\sqrt{c}(N-k)^2}{N^{1.5}} \mu_k^2 }
	{(N-k)\mu_k^2}
	\leq
	 \frac{\sqrt{c}(N-k)}{N^{1.5}}.
\end{equation*}
\end{proof}

{\bf Remark 3.} The bound on the relative error on the estimator is pessimistic.
In fact, using the techniques in Lemmas \ref{lemma:DeltaNonHerm} and \ref{lemma:DeltaHerm}
we could prove that the upper bound is ${\cal O}(1/N)$, which also agrees with 
experimental observations. However, such complexity is unnecessary as our goal 
is simply to show that our model for $\Var(t(A_R)))$ is sufficiently accurate for large $N$.

{\bf Remark 4.} By setting $k=0$ in Theorem \ref{thm:exp-vars} we obtain 
expressions for $E(Var(t(A)))$, the undeflated Hutchinson estimator:
$(N-1)(V_0+\mu_0^2)$ for non-Hermitian $A$ and 
$N^2 V_0/(N+1)$ for Hermitian $A$.

{\bf Remark 5.} Our original assumption that the singular vector matrices
  are random unitary is not required by Theorem \ref{thm:exp-vars}. 
It is sufficient that the elements of $\Delta_{ml}$ have expectation 
  values and variances as given by Lemmas \ref{lemma:DeltaNonHerm} and 
  \ref{lemma:DeltaHerm}.

\begin{corollary}
For non-Hermitian matrices and for any $1\leq k \leq N$, 
  $E(\Var(t(A_R))) \leq E(\Var(t(A)))$.
For Hermitian matrices, the expected deflated variance reduces only if 
$\mu_0^2  - \frac{(N-k)^2}{N^2} \mu_k^2 < \frac{1}{N} \sum_{i=1}^k \sigma_i^2$.
\label{Corollary:ratios}
\end{corollary}
\begin{proof}
Based on Theorem \ref{thm:exp-vars} and Remark 4, we want the ratio
  of deflated to undeflated variance
$$
   \frac{E(\Var(t(A_R)))}{E(\Var(t(A)))} 
       = 
   \frac{ (N-k) (1-\frac{1}{N}) (V_k + \mu_k^2) }
       {(N-1)(V_0+\mu_0^2)}
       = 
   \frac{(N-k) (V_k + \mu_k^2)}
       {  N   (V_0+\mu_0^2)} \leq 1.$$
Note that $V_0 +\mu_0^2 =  \frac{1}{N} \sum_{i=1}^N \sigma_i^2$.
  and $V_k +\mu_k^2 =  \frac{1}{N-k} \sum_{i=k+1}^{N} \sigma_i^2$.
Because $\sigma_m \geq \sigma_l,\ m\leq k, l>k$,
  the expected value of their squares will also be larger and thus
  $V_0 +\mu_0^2 > V_k +\mu_k^2$, which proves the inequality.

For Hermitian cases the ratio of deflated to undeflated variance becomes
$$
   \frac{E(\Var(t(A_R)))}{E(\Var(t(A)))}
   	=
   \frac{ (N-k) (N V_k + k \mu_k^2) } {N V_0}
       = 
   \frac{\sum_{i=k+1}^N \sigma_i^2 -\frac{(N-k)^2}{N} \mu_k^2}
       { \sum_{i=1}^N \sigma_i^2 - N \mu_0^2}.
$$
Requiring that the above ratio is less than one yields the desired result.
\end{proof}

Corollary \ref{Corollary:ratios} shows that for non-Hermitian matrices 
  the condition (\ref{eq:improvedPessimistic}) on the decay of the 
  singular values is unnecessary---deflation will always reduce the variance.
This is not guaranteed for Hermitian matrices, for which 
  condition (\ref{eq:improvedPessimistic}) seems to be valid, as we show
  experimentally in the next section.

Although the above corollary is qualitative,
Theorem \ref{thm:exp-vars} facilitates a quantitative prediction of the outcome of deflation
  based solely on the variance and the expectation of the undeflated singular values.
Often, users have some idea of the singular spectrum of their matrix and thus can decide 
  not only if deflation works, but also how many singular triplets to deflate and what 
  the expected benefit will be.
Moreover, as we show in our numerical experiments, the estimates based on our formulas
  are extremely robust even when the eigenvectors are not random unitary matrices.

\section{The effect of the singular spectrum}
Having factored out the effects of the singular vectors, we now study the 
  effect of the singular value distribution using the previous theory to 
  predict actual experiments.
Clearly, the larger the gap between deflated and undeflated singular values,
 the larger the reduction in (\ref{eq:vA-vAR}).
In the following experiments we study the effect of deflation for six 
  different model distributions of $\sigma_i$.

Given a diagonal matrix of singular values $\Sigma$, we generate a pair 
  of random unitary matrices $U$ and $V$, and construct 
  one Hermitian matrix $U\Sigma U^H$ and one non-Hermitian matrix
  $U\Sigma V^H$.
For each model distribution we construct matrices of several sizes.
We report the ratio of the variance of the deflated matrix, where we deflate
  various percentages of its largest singular triplets, to the variance of 
  the original undeflated matrix.
This can be computed explicitly from (\ref{eq:variance}),
  or through our model in Theorem \ref{thm:exp-vars}.
As statistically expected, beyond small matrices of dimension less than 100,
  there is perfect agreement between our model predictions and experimentally 
  determined variances.
Thus, we only present results from our model.

\begin{figure}[h]
  \centering
  \subcaptionbox{log\label{log_c2_multiN}}{\includegraphics[width=0.45\textwidth]{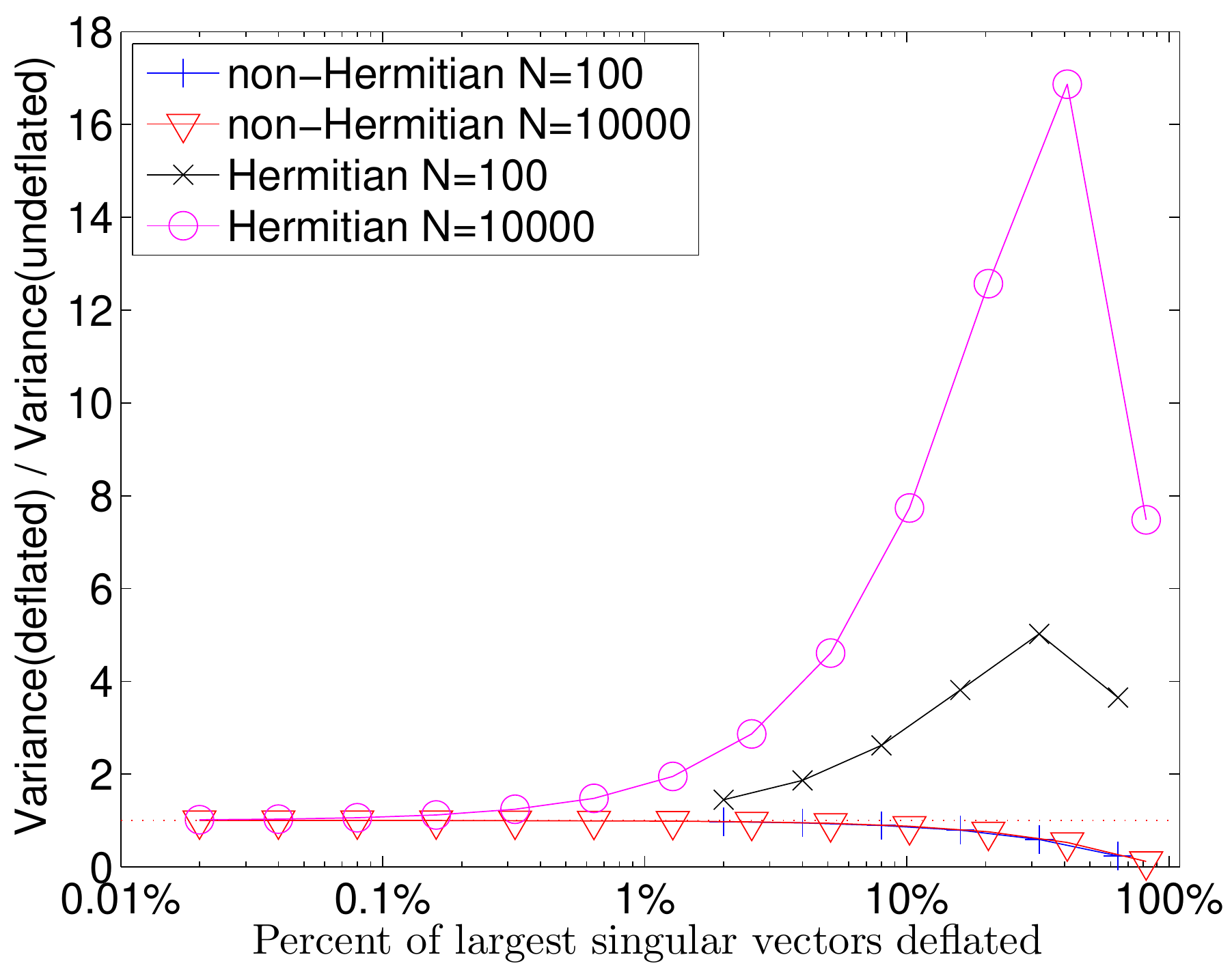}}\hspace{0em}%
  \subcaptionbox{square root\label{sqrt_multiN}}{\includegraphics[width=0.45\textwidth]{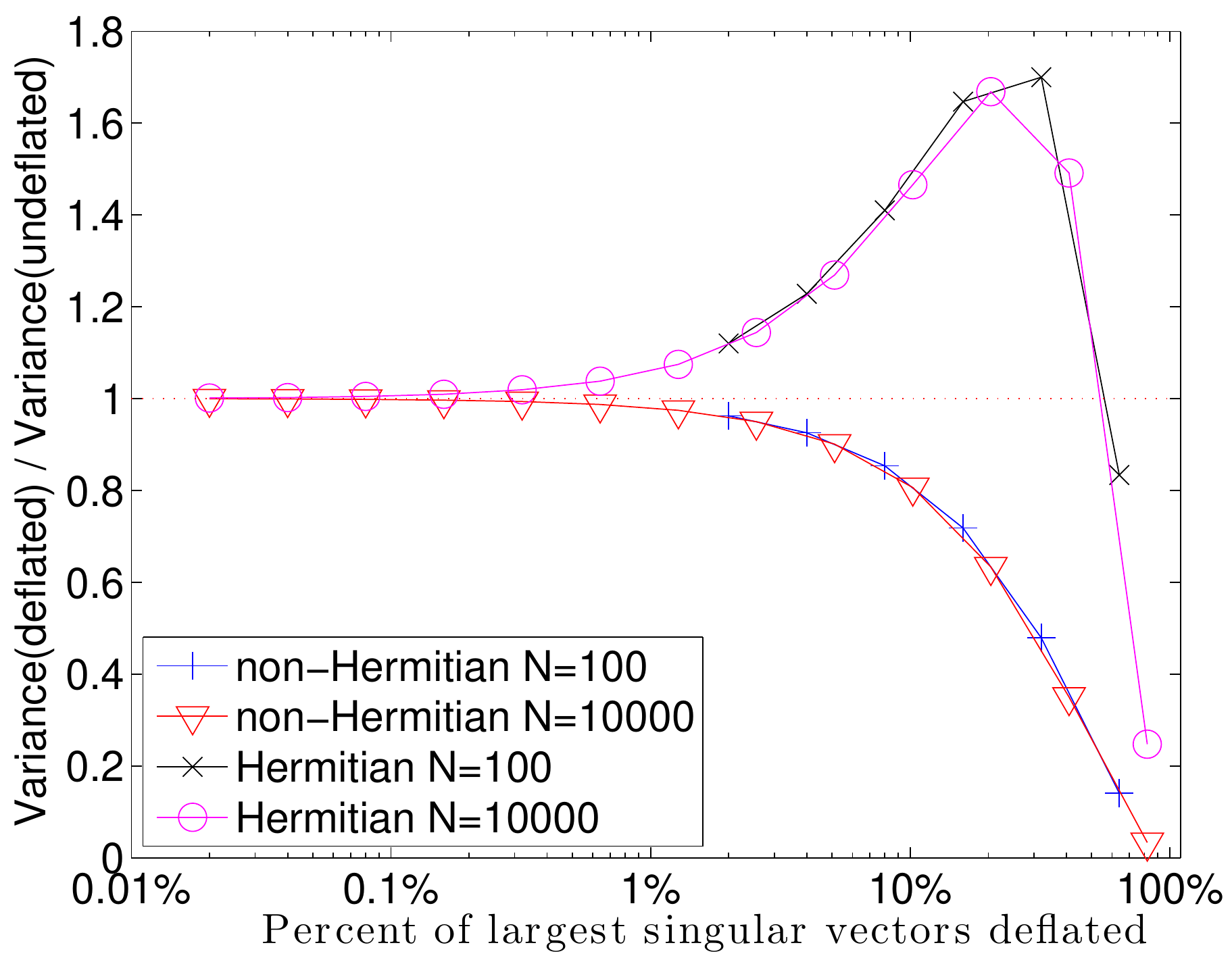}}
  \caption{On the left is a logarithmic spectrum: $\sigma_{N-i+1} = 1 + 2\cdot\operatorname{log}(i)$. On the right is a square root spectrum: $\sigma_{N-i+1} = \sqrt{i}$. The dotted red line in both plots is a constant line at $y=1$. Points below this line signify an improvement in variance with deflation. Points above the line denote a deflated operator with a higher Frobenious norm than the original matrix, a case in which deflation is hurtful and variance increases.}
\end{figure}

\begin{figure}[h]
  \centering
  \subcaptionbox{linear\label{linear_multiN}}{\includegraphics[width=0.33\textwidth]{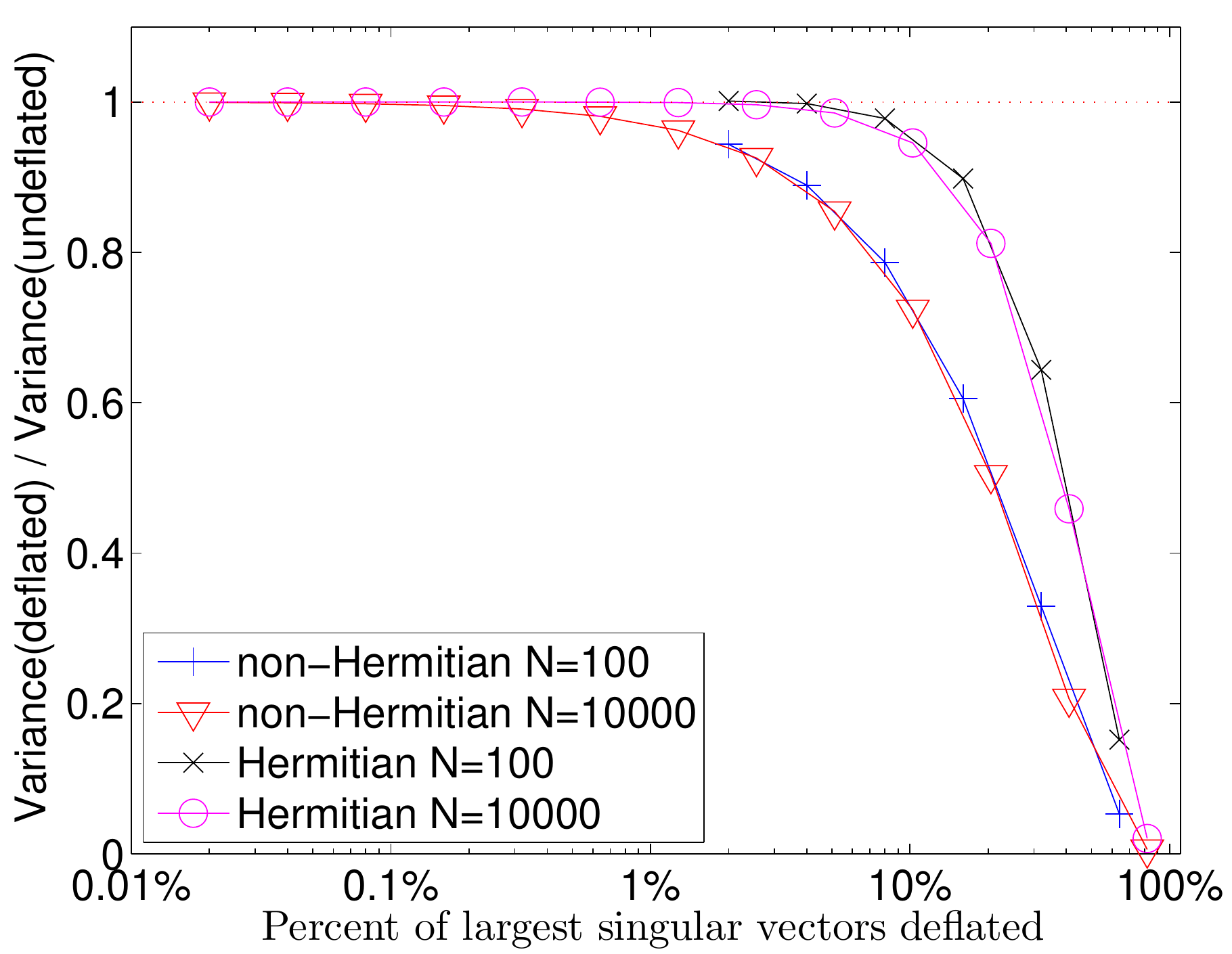}}
  \subcaptionbox{quadratic\label{quadratic_multiN}}{\includegraphics[width=0.33\textwidth]{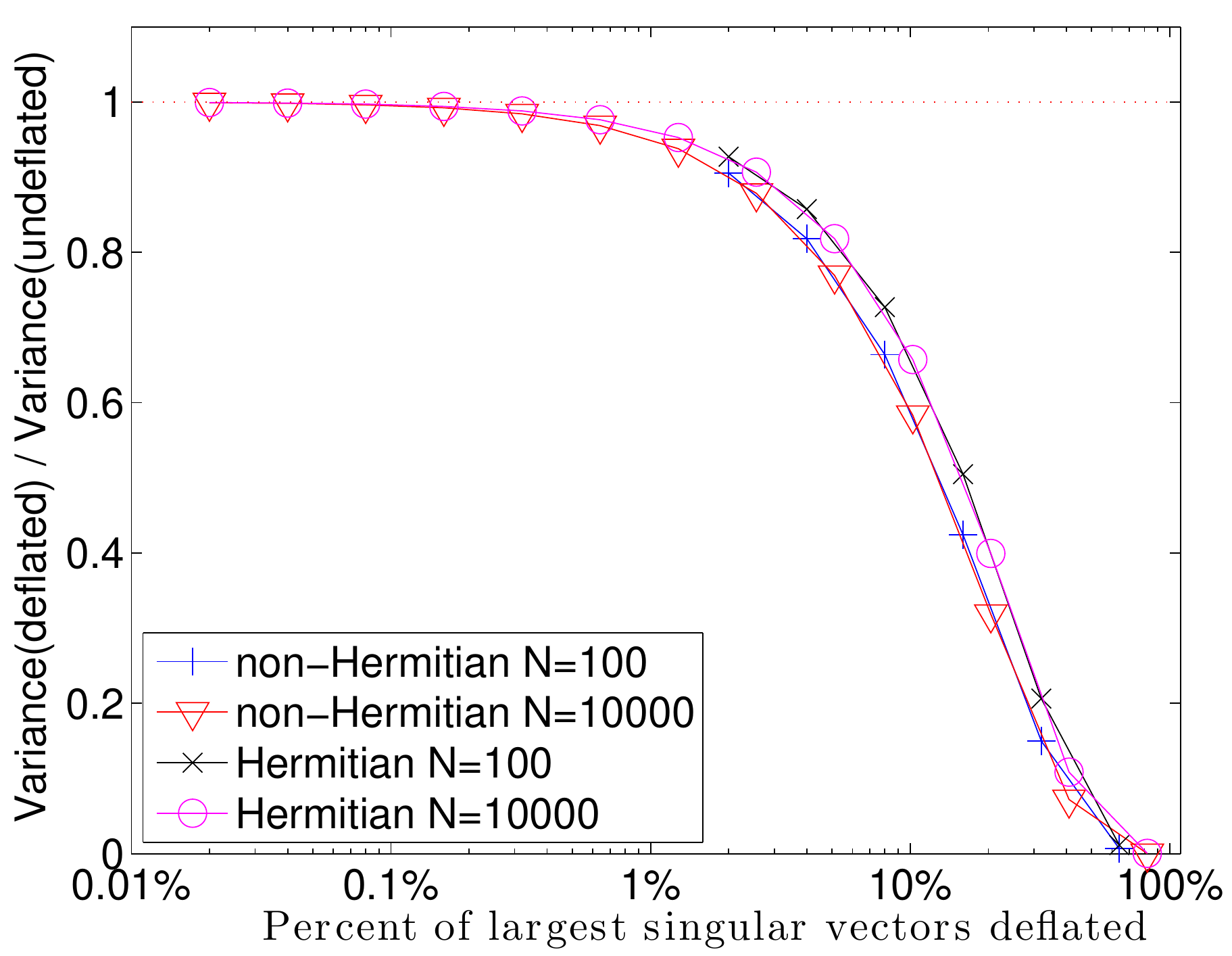}}\hspace{0em}%
  \subcaptionbox{cubic\label{cubic_multiN}}{\includegraphics[width=0.33\textwidth]{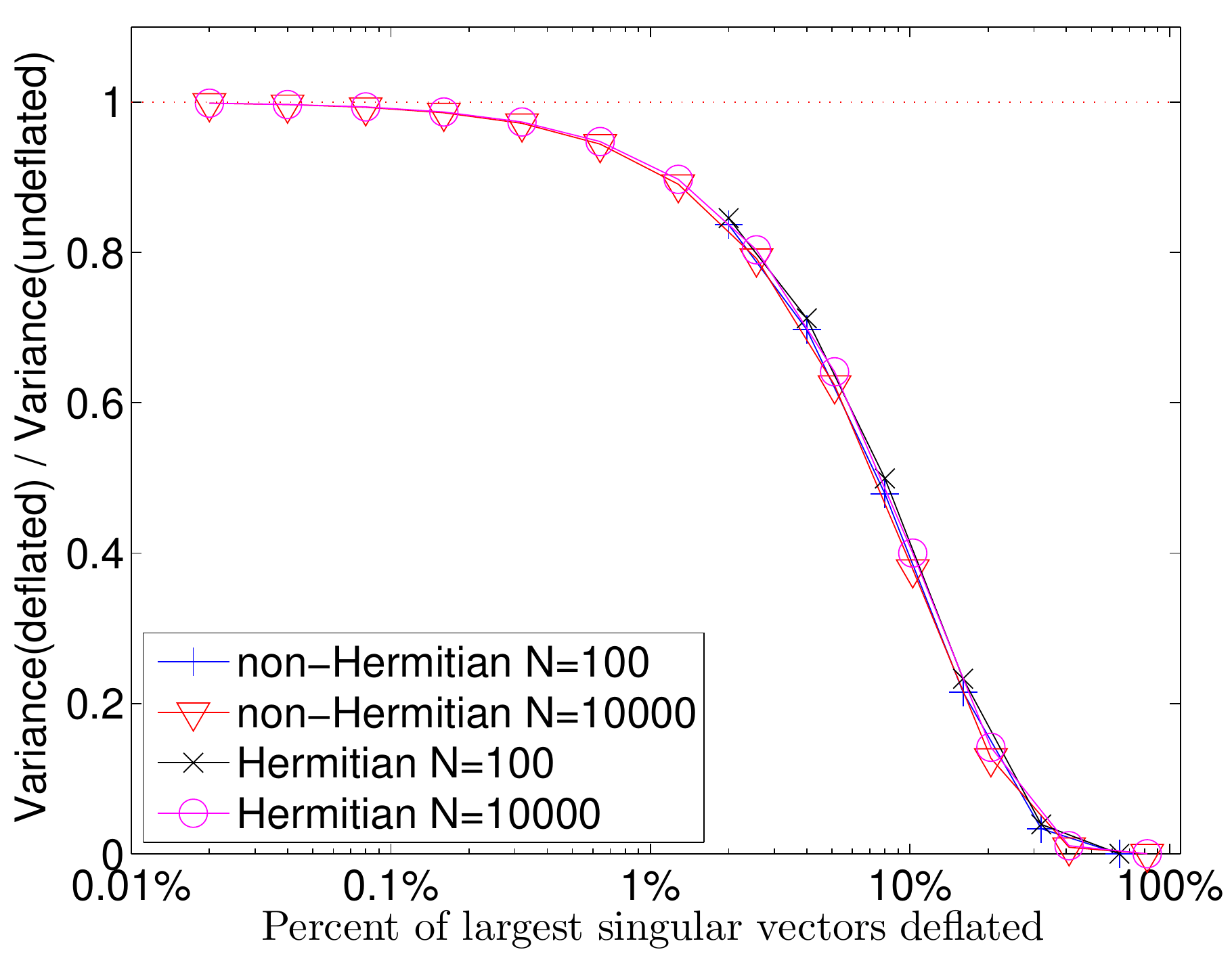}}
\caption{Variance reduction ratios for matrices with spectra with 
linear, $\sigma_{N-i+1} = i$,
quadratic, $\sigma_{N-i+1} = i^2$,
and cubic, $\sigma_{N-i+1} = i^3$, growth rates. }
\end{figure}

In Figure \ref{log_c2_multiN} we consider a model where the singular values
  increase at a logarithmic rate with respect to their index. 
As Corollary \ref{Corollary:ratios} predicts, for Hermitian matrices the 
  variance increases with the number of deflated singular triplets, 
  and the problem is more pronounced with larger matrix size. 
Although for non-Hermitian matrices the ratio is always below one,
  it requires deflating a substantial part of the spectrum to 
  reduce the variance appreciably.
In Figure \ref{sqrt_multiN} the spectrum increases as the square root 
  of the index, and the effects of deflation, although improved,
  still are not beneficial for Hermitian matrices.

In Figures \ref{linear_multiN}, \ref{quadratic_multiN}, and \ref{cubic_multiN}
  the growth of the singular values is linear, quadratic, and cubic, 
  respectively.
The ratio is now below one for both types of matrices (confirming 
  the condition (\ref{eq:improvedPessimistic}) for Hermitian matrices).
We can see that with larger growth rates, the variance reduction is larger 
  for a particular fraction of singular values deflated.
Additionally, for sufficiently large growth rates, the difference between 
  Hermitian and non-Hermitian matrices vanishes.

\begin{figure}[h]
  \centering
  \subcaptionbox{$1/\sqrt{i}$\label{oneOverSqrt_multiN}}{\includegraphics[width=0.45\textwidth]{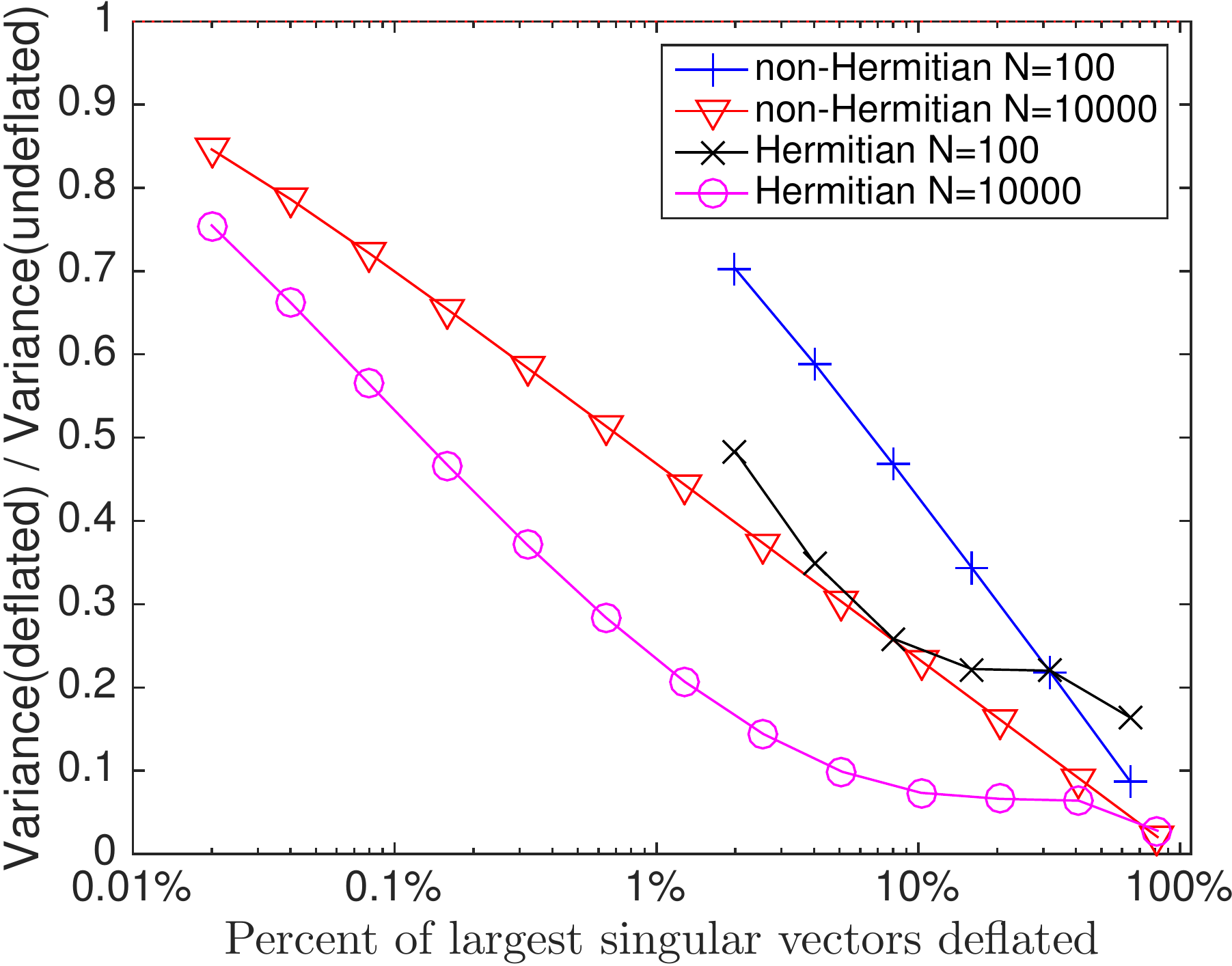}}\hspace{0em}%
  \subcaptionbox{Laplacian\label{2D_Laplacian_Inverse_multiN}}{\includegraphics[width=0.46\textwidth]{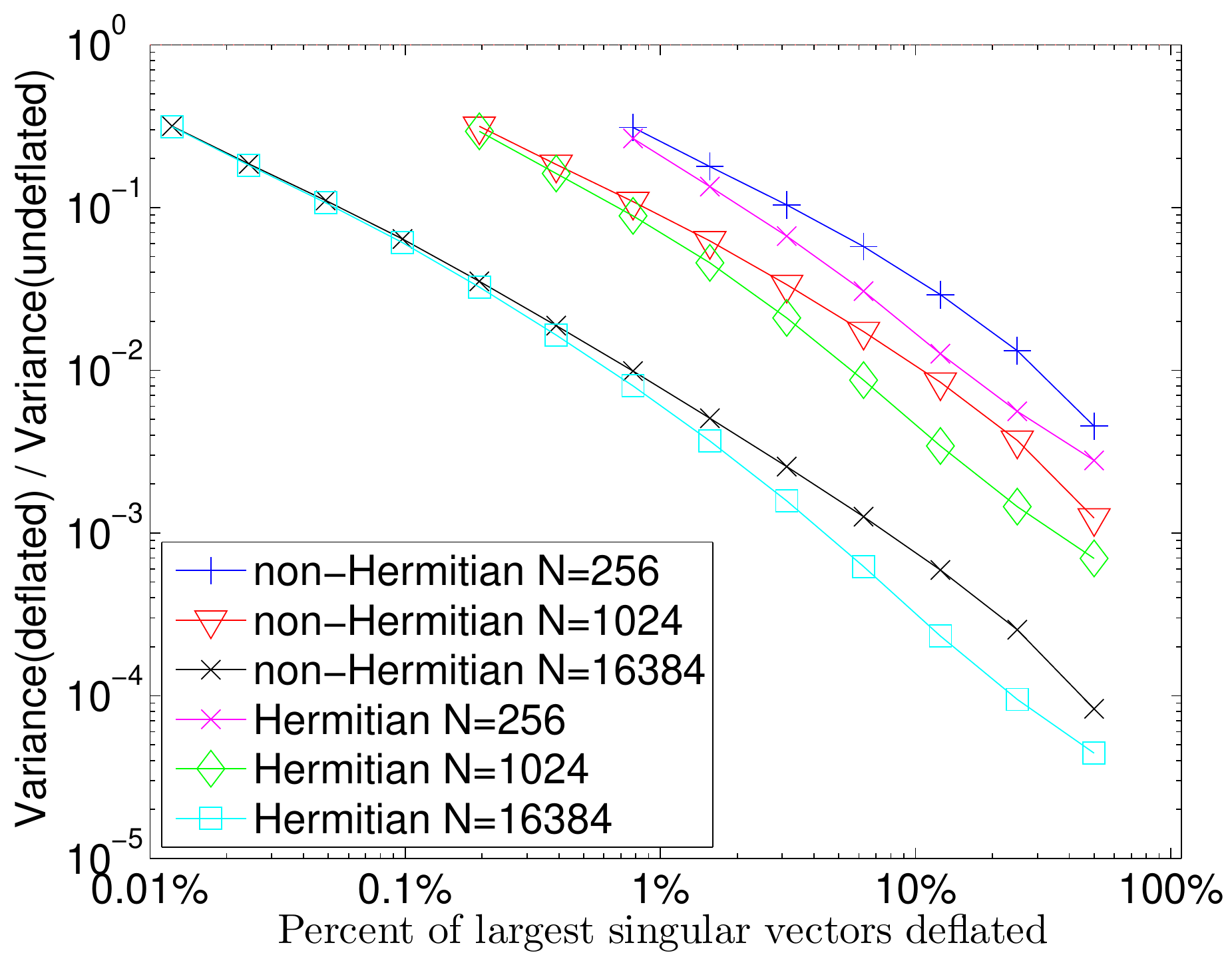}}
  \caption{On the left we have spectrum $\sigma_{i} = 1/\sqrt{i}$. The right plot shows the deflation of the inverse of a 2D discrete Laplacian on a grid $\sqrt{N}\times\sqrt{N}$ with Dirichlet boundary conditions.}
\end{figure}

For spectra that decay as a rational polynomial, the picture is different.
Figure \ref{oneOverSqrt_multiN} shows an example where the spectrum
  is $\sigma_i = 1/\sqrt{i}$.
There are a few large singular values but the rest do not reduce appreciably.
The effect of this is that Hermitian matrices experience larger relative 
  improvement with deflation over non-Hermitian matrices. 
We have observed this effect also for other rational polynomials, $1/i^p$, 
  but the difference between matrix types seems to peak at the $1/\sqrt{i}$.
This observation is particularly relevant to our problem of finding 
  the trace of the inverse of a matrix.
In Figure \ref{2D_Laplacian_Inverse_multiN} we study the spectrum 
  of the inverse of the discrete Laplacian, a common problem that also 
  has some of the features of our target QCD problem at the free field limit.
We see significant variance reduction, especially as the lattice size grows.
The Hermitian matrices continue to have an advantage over non-Hermitian
  matrices, but the difference for practical problems is negligible.

\FloatBarrier

\section{Experiments on general matrices}
The previous section studied the effect of spectra on matrices with 
  random unitary singular vectors. 
In this section we investigate the extent to which our theory is applicable 
  to general matrices with singular vectors that are not random.
We choose four matrices from the University of Florida sparse matrix 
  collection \cite{Davis:2011:UFS:2049662.2049663} with relatively small sizes (675--2000)
  that are derived from real world problems in various fields, such 
  as chemical transport modeling and magnetohydrodynamics. 
In all our results we deflate the estimator of the trace of $A^{-1}$.

\begin{figure}[h]
  \centering
  \subcaptionbox{BWM2000\label{BWM2000_Inverse_Experimental}}{\includegraphics[width=0.45\textwidth]{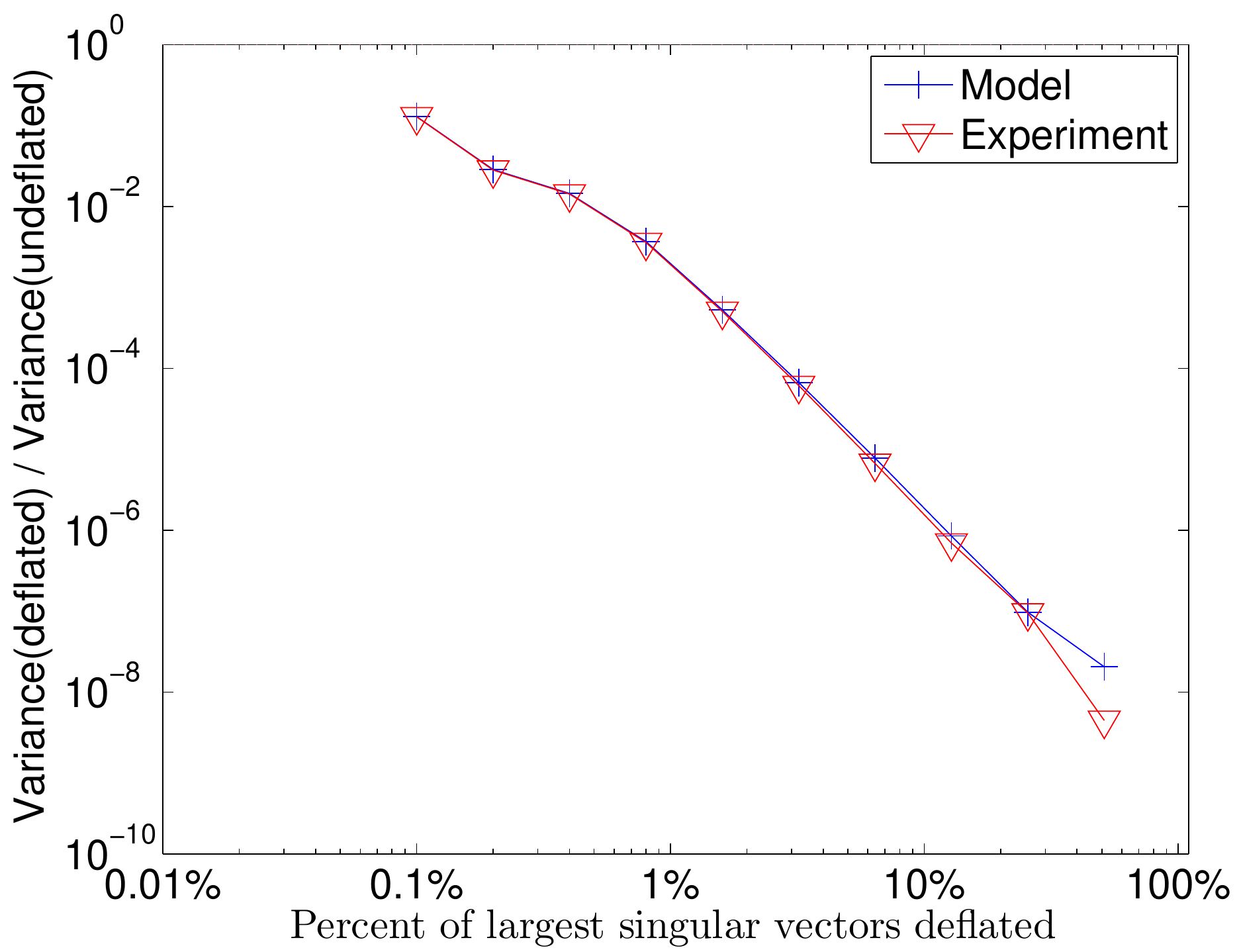}}\hspace{0em}%
  \subcaptionbox{MHD1280B\label{MHD1280B_Inverse_Experimental}}{\includegraphics[width=0.45\textwidth]{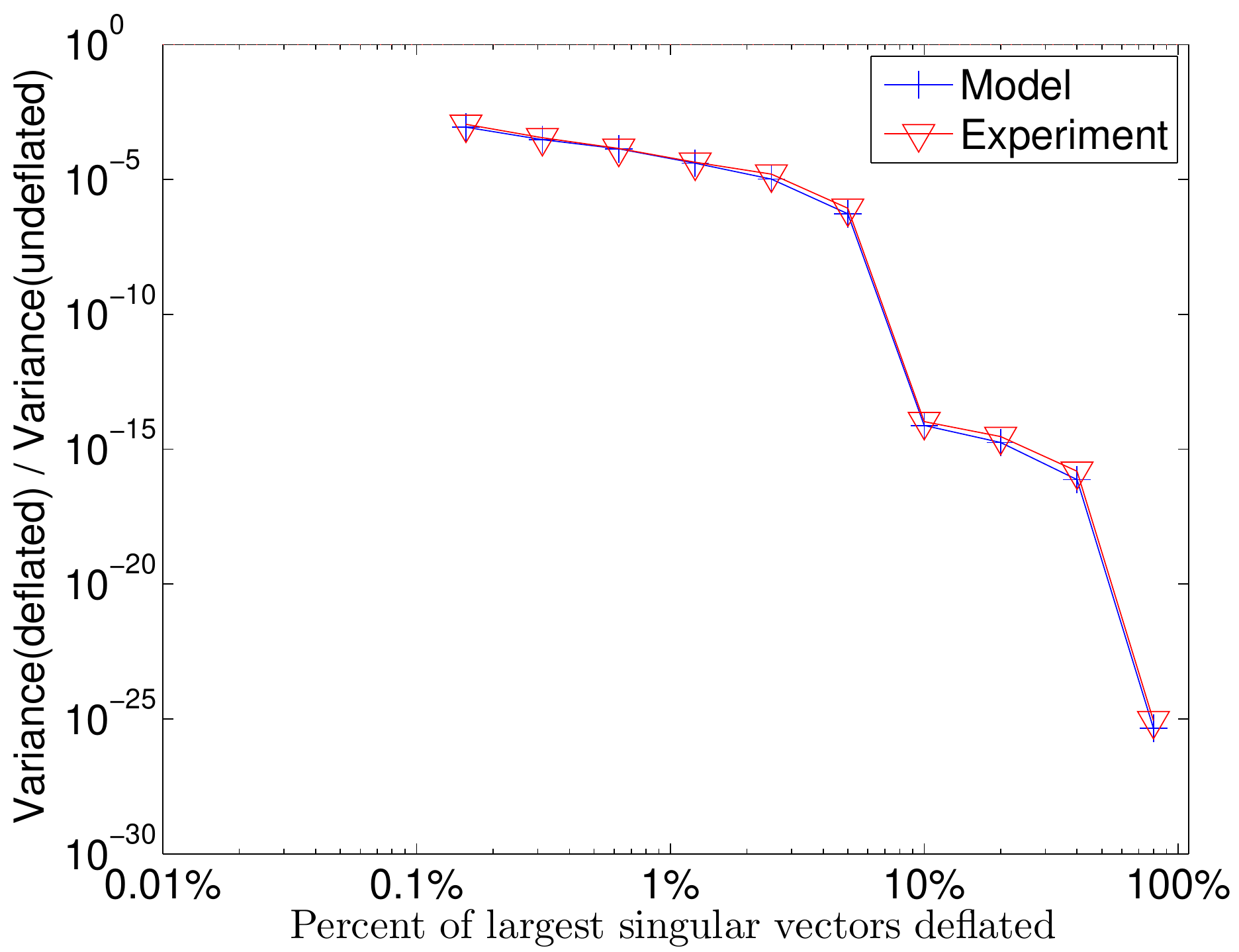}}
  \caption{Matrix BWM2000 has a size of $N=2000$ and condition number of 2.37869e+5. Matrix MHD1280B has $N=1280$ and a condition number of 4.74959e+12. Both matrices are real, non-symmetric.}
\end{figure}

\begin{figure}[h]
  \centering
  \subcaptionbox{NOS6\label{NOS6_Inverse_Experimental}}{\includegraphics[width=0.45\textwidth]{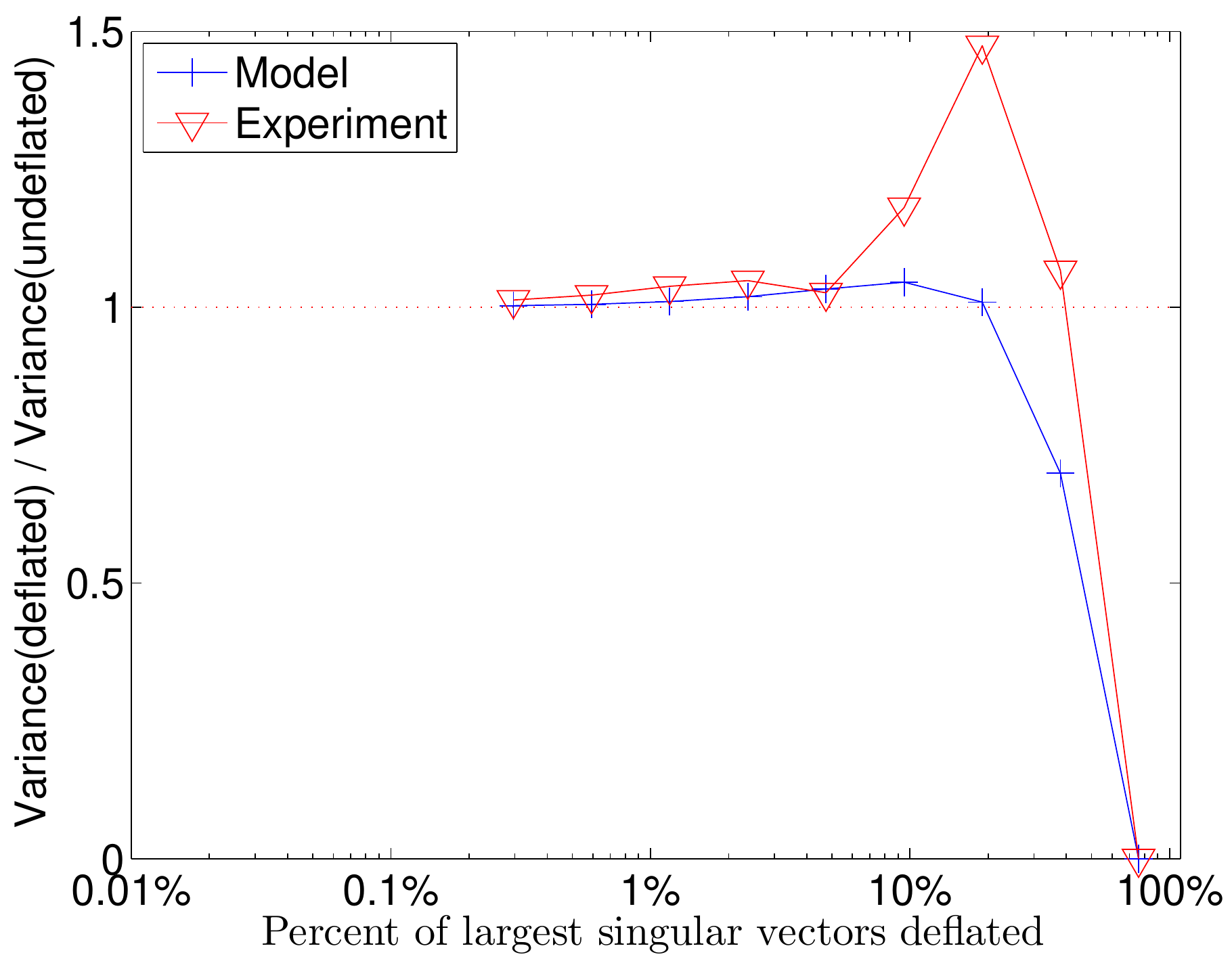}}\hspace{0em}%
  \subcaptionbox{OLM1000\label{OLM1000_Inverse_Experimental}}{\includegraphics[width=0.45\textwidth]{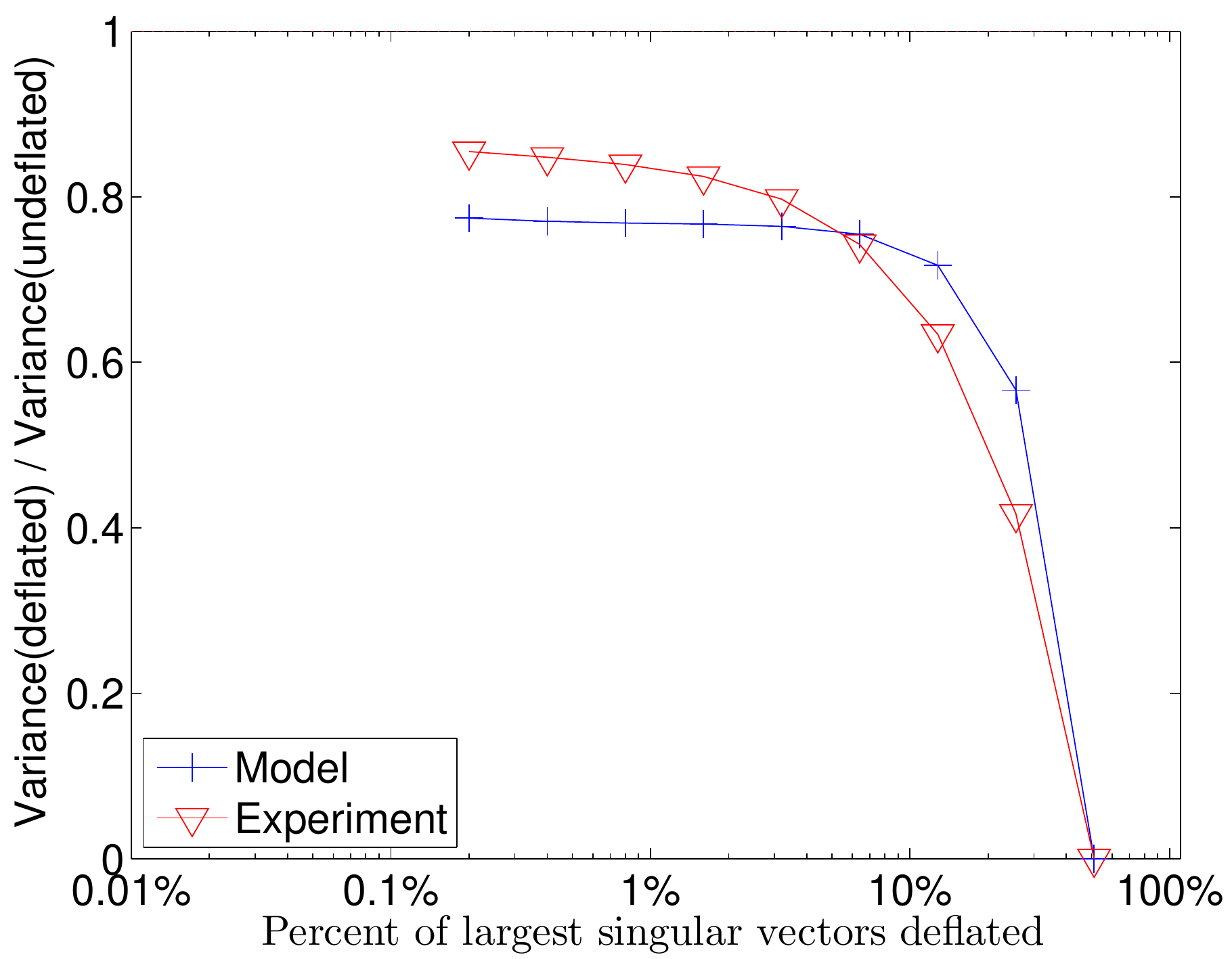}}
  \caption{Matrix NOS6 is symmetric, with $N=675$ and condition number of 7.65049e+06. Matrix OLM1000 is non-symmetric, with $N=1000$ and a condition number of 1.48722e+06. }
\end{figure}
For Figures \ref{BWM2000_Inverse_Experimental} and 
  \ref{MHD1280B_Inverse_Experimental}, the model and experimental results 
  agree very closely. 
Both demonstrate dramatic variance reduction even when deflating a small 
  fraction of the SVD space.
Both matrices have a high condition number implying that their 
  small singular values contribute most of the variance of the matrix inverse
  estimator.
Thus it pays to remove them.

In contrast, deflation does not improve the variance for the matrix in 
  Figure \ref{NOS6_Inverse_Experimental}, unless almost the entire spectrum
  is deflated.
For deflating less than 10\% of the singular triplets---the most realistic situation---model and experimental results agree.
Beyond that number, the experiment performs worse than predicted. 
However, the model still captures the overall effect and recommends avoiding 
  deflation altogether.
In Figure \ref{OLM1000_Inverse_Experimental} the effect of deflation is
  beneficial but limited. 
The disagreement between model and experiment is about 10\%, and therefore 
  the model can be used to predict the outcome effectively.

In summary, the presence of non random singular vectors could generate 
  a few discrepancies, but these and other extensive experiments show that 
  our model is useful in predicting the overall effect of deflation. 
Specifically, even a rough knowledge of the particular singular value spectrum
  can help us determine whether deflation would be valuable or hurtful.
Finally, we emphasize the small sizes of the above matrices. 
In large real world problems, the singular vectors are more likely 
  to behave like random ones.

\section{Application of deflation to Lattice QCD}

In LQCD, we may assume that the singular vectors of the Dirac operator are 
  approximately random, uniformly distributed unitary matrices.
This is justified by the random matrix theory~\cite{Shuryak:1992pi,Verbaarschot:2000dy} approximation to QCD.
 In this approximation, 
  the Dirac matrix is replaced by a random matrix with a suitable 
  probability distribution that satisfies the fundamental symmetries of QCD.
It has been shown that this approach explains the numerically observed 
  spectral density of the Dirac matrix very well \cite{Verbaarschot:2000dy}.
Therefore, we expect that our model should capture the essential properties of deflation on the stochastically estimated trace
\footnote{Using a non-uniform distribution of the singular vector matrix, such as the distributions used in~\cite{Carlsson:2008dh}, similar results can be obtained.}.

\subsection{How to obtain the deflation space}
We are interested in the trace of the inverse of a matrix, so
  we need to compute its smallest singular triplets.
Then, we apply the Hutchinson method on the deflated matrix 
  by solving a series of linear systems of equations.
It would have been desirable to compute the deflation 
  space from the search spaces built by the iterative methods for 
  solving these linear systems.
This idea has been explored effectively for Lattice QCD in the past
  \cite{Morgan_Wilcox,Stathopoulos:2007zi,AbdelRehim:2009by}.
However such methods are not suitable for our current problem for 
  the following reasons.

First, methods such as GMRESDR or eigBICG produce approximations 
  to the lowest magnitude eigenvalues of the non Hermitian matrix $A$. 
Much experimentation has shown that this eigenspace is not effective
  as deflation for reducing the variance of the Hutchinson method. 
To produce the smallest singular triplets we would have to work with 
  eigCG on the normal equations $A^HA$ \cite{Stathopoulos:2007zi}.
Second, only the lowest few eigenpairs produced by eigCG are accurate. 
The rest may have a positive effect on speeding up the linear solver,
  but they do not seem adequate for variance reduction.
Third, and most important, we are interested in large scale 
  problems for which unpreconditioned eigCG would not converge
  in reasonable time. 
However, if a preconditioner $M^{-1}$ is used, all the above methods find
  the eigenpairs of $M^{-1}A$ or of $M^{-H}M^{-1} A^HA$.
These may help speed up the linear solver but are not relevant for
  deflating $A^{-1}$ for variance reduction.

The alternative is to compute the deflation space through an explicit 
  eigensolver on $A^HA$. 
Numerical difficulties arising by using the $A^HA$ are not an issue for 
  the relatively low accuracy needed for deflation.
This is a challenging problem for our large problem sizes because the 
  lower part of the spectrum becomes increasingly dense and 
  eigenvalue methods converge slowly.
Moreover, to achieve sufficient deflation power, hundreds or 1000s of 
  eigenpairs must be computed.
Although Lanczos type methods are good for approximating large parts 
  of the spectrum, they cannot use preconditioning so they are unsuitable 
  for our problems.

We have used the state-of-the-art library PRIMME (PReconditioned Iterative 
  MultiMethod Eigensolver) \cite{PRIMME} which offers 
  a suite of near-optimal methods for Hermitian eigenvalue problems.
Among several unique features, PRIMME has recently added support for 
  solving large scale SVD problems, including preconditioning capability,
  something that is not directly supported by other current software.
In Lattice QCD, a multi-group, multi-year effort has resulted in 
  a highly efficient preconditioner which is based on domain decomposition 
  and adaptive Algebraic Multigrid (AMG) \cite{Babich:2010qb}.
In that community, AMG is a game changer but it has only been used 
  to solve linear systems of equations. 
We employ AMG as a preconditioner in PRIMME to find 1000 lowest singular 
  triplets. 
For most methods in PRIMME, AMG accelerates the number of iterations by 
  orders of magnitude and results in wallclock speedups of around 30. 

To obtain the best performance for our problem we have experimented
  with various PRIMME methods and parameters, and AMG configurations. 
A determining factor for these optimizations was the accuracy with which 
  eigenvectors had to be computed.
For high accuracy, PRIMME's near-optimal method GD+k is the method of choice,
  but for the low accuracy that is sufficient for variance reduction, 
  i.e., residual tolerance less than 1e-2 or 1e-3, methods with a 
  large block size are typically more efficient.
There are a couple of reasons for this, besides better cache utilization 
  and lower memory traffic.
Single vector methods with large tolerance may misconverge to interior
  eigenvalues before the exterior ones become visible to the method. 
In our problem, the smallest few eigenvalues are $O$(1e-5) so tolerance 
  has to be smaller than that.
Large block size avoids this problem and, additionally, allows many 
  eigenpairs to converge to much lower residual norms than the 
  requested tolerance. 
We experimented with various block sizes and methods and settled to
  GD+k with block-size of 30 and a total subspace of 90.
This is equivalent to the LOBPCG method with 90 vectors that locks 
  converged eigenvectors out of the basis as they converge.
This window approach is far more efficient than the original LOBPCG, 
  which sets the block size equal to the number of required eigenvalues.
In PRIMME this method can be called directly as 
  \texttt{LOPBCG\_Orthobasis\_Window}.

The AMG software provides a solver for a non-Hermitian linear system 
  $Ax = b$, not just a preconditioner. 
There are three levels of multigrid with a GCR smoother at each level 
  \cite{Babich:2010qb}.
Because PRIMME needs a preconditioner for the normal equations, 
   $A^HA \delta = r$, 
  each preconditioning application involves two calls to AMG to solve 
  the two systems approximately, $A^Hy = r$ and $A\delta = y$.
We found 4 GCR iterations at the fine level and 5 GCR iterations at each 
  of the two coarse levels to be optimal.
Preconditioning for eigenvalue problems differs from linear systems in 
  the sense that it should approximate $(A^HA - \sigma I)^{-1}$
  to improve eigenvalues near $\sigma$. 
In our AMG preconditioner $\sigma$ is zero, so we expect the quality of the 
  preconditioner to wane as we find eigenvalues inside the spectrum.
However, the lowest part of the spectrum is quite clustered and as such
  for multigrid this deterioration is small.

As we discuss in the experiments section, our code was able to efficiently 
  produce one thousand eigenpairs in one of the largest eigenvalue 
  calculations  we performed in Lattice QCD.

\subsection{Deflating the trace method and combining with Hierarchical Probing}

Given $k$ eigenpairs ($\Lambda, V$) of the normal equations,
  the left singular vectors can be obtained as $U = AV\Sigma^{-1}$, 
  where $\Sigma = \Lambda^{1/2}$.
Following (\ref{eq:deflatedA}), we can decompose
$\Tr(A^{-1}) = \Tr(A^{-1}_D)+\Tr(A^{-1}_R) = 
	\Tr(V\Sigma^{-1}U^H) + \Tr(A^{-1}-V\Sigma^{-1}U^H)$.
Using the cyclic property of the trace, we have
$\Tr(V\Sigma^{-1}U^H) = \Tr(\Sigma^{-1}U^HV) = \Tr(\Lambda^{-1}V^HA^HV)$.
This means that the trace of $A^{-1}_D$ can be computed explicitly
  through $k$ matrix vector multiplications and $k$ inner products.
Similarly, we see that $\Tr(A^{-1}_R) = \Tr(A^{-1}-V\Lambda^{-1}V^HA^H)$,
  so the quadratures required in Hutchinson's method can be 
  computed as $z^HA^{-1}z$ and $z^HV \Lambda^{-1} V^H (A^Hz)$.
This means that we can avoid the significant storage of $U$.

We now have all the components to run the deflated Hutchinson method
  using random Rademacher vectors. 
However, the same deflation technique can be used on the Hutchinson method 
  if the vectors come from the Hierarchical Probing (HP) method. 
HP uses an implicit distance-$d$ coloring of the lattice
  to pick the probing vectors as certain permutations of Hadamard vectors
  that remove all trace error that corresponds to $A_{ij}^{-1}$ elements
  with $i,j$ having up to $d$ Manhattan distance in the lattice.
The hope is that deflation removes error in a complementary way from HP
  and the two techniques together lead to faster convergence. 

To avoid the deterministic bias of the HP method, we follow the technique
  proposed in \cite{Stathopoulos:2013aci} which first computes a random 
  \textcolor{black}{$\mathbb{Z}_4$ vector} $z_0$, and then in the Hutchinson method uses the 
  vector $z = z_0 \odot z_h$, which is the elementwise product of $z_0$ 
  with each Hadamard vector $z_h$ from the HP sequence. 
We have shown this method to be unbiased and to reduce the measured error.
Algorithm \ref{Alg:main} summarizes our approach.

\begin{algorithm}
\caption{$Trace = $ deflatedHP$(A)$ }
\begin{algorithmic}[1]
\State $[\Lambda, V] = \operatorname{PRIMME}(A^HA)$ 
\State $T_D = \Tr(\Lambda^{-1}V^HA^HV)$; $T_R = 0$
\State \textcolor{black}{$z_0=\operatorname{randi}([0,3], N, 1)$; $z_0=exp(z_0 \pi i/2)$}
\For{$j=1:s$}
\State $z_h = $ next vector from Hierarchical Probing or other scheme
\State $z = z_0 \odot z_h$
\State Solve $Ay = z$
\State $T_R = T_R + z^Hy - z^HV \Lambda^{-1} V^H (A^Hz)$
\State $Trace = T_R/j + T_D$
\EndFor
\end{algorithmic}
\label{Alg:main}
\end{algorithm}

We conclude this algorithmic part of the paper by mentioning an important 
  application of this technique.
In Lattice QCD, we are often interested in computing $\Tr(\Gamma A^{-1})$ 
  for several different 
  $\Gamma$ matrices whose application to a vector are inexpensive to compute.
In such cases, the SVD decomposition (\ref{eq:deflatedA}) still applies,
$\Tr(\Gamma A^{-1}) = 
	\Tr(\Gamma V\Lambda^{-1}V^HA^H) + 
	\Tr(\Gamma A^{-1}-\Gamma  V\Lambda^{-1}V^HA^H)$.
The computations are similar to Algorithm \ref{Alg:main}, with a $\Gamma$ 
  matrix vector product inserted at each step. 
Therefore, the computational cost of the SVD and the storage for singular 
  vectors can be amortized 
  by reusing the deflation space to compute traces with multiple $\Gamma$ matrices.
  
\section{QCD Experiments}
We present results from experiments with two representative Dirac matrices.   Both are from $32^3\times 64$, $\beta=6.3$ Clover improved Wilson ensembles. In both cases, the pion mass was about $300 MeV$.  However, the first matrix comes from an ensemble with 3 flavors of dynamical quarks, whose masses were turned to match the physical strange quark mass. In this case we employed a lower quark mass (quark mass $m_q=-0.250$ in lattice units) for our numerical experiments in order to achieve a more singular matrix. In the second case, the ensemble from which we selected the Dirac matrix is one with 2 light quark flavors and one strange quark. The strange quark is again, at its physical value, and the light quarks have masses $-0.239$ that result in 300MeV pions.  The interested reader can find further  details about these ensembles in~\cite{Yoon:2016dij}. Subsequently, we will refer to the matrix with a quark mass of $m_q=-0.250$ as the Dirac operator from ensemble A, and the $m_q=-0.239$ mass matrix as the Dirac matrix from ensemble B.

The above matrices have a size of $N=$ 25,165,824 and condition numbers of 1747 and 1788 respectively.  
The subspaces were obtained using PRIMME set to the 
  {\tt LOBPCG\_Orthobasis\_Window} method with a tolerance of $10^{-2}$ and a block size of 30 \cite{PRIMME}. This was supplemented with a three level AMG preconditioner with $4^4$ and $2^4$ blocking and a fine/coarse maximum iteration count of 4 and 5 respectively \cite{Babich:2010qb}.

\begin{figure}[h]
  \centering
  \subcaptionbox{QCD SVD spectra\label{SVDSpectrum}}{\includegraphics[width=0.49 \textwidth]{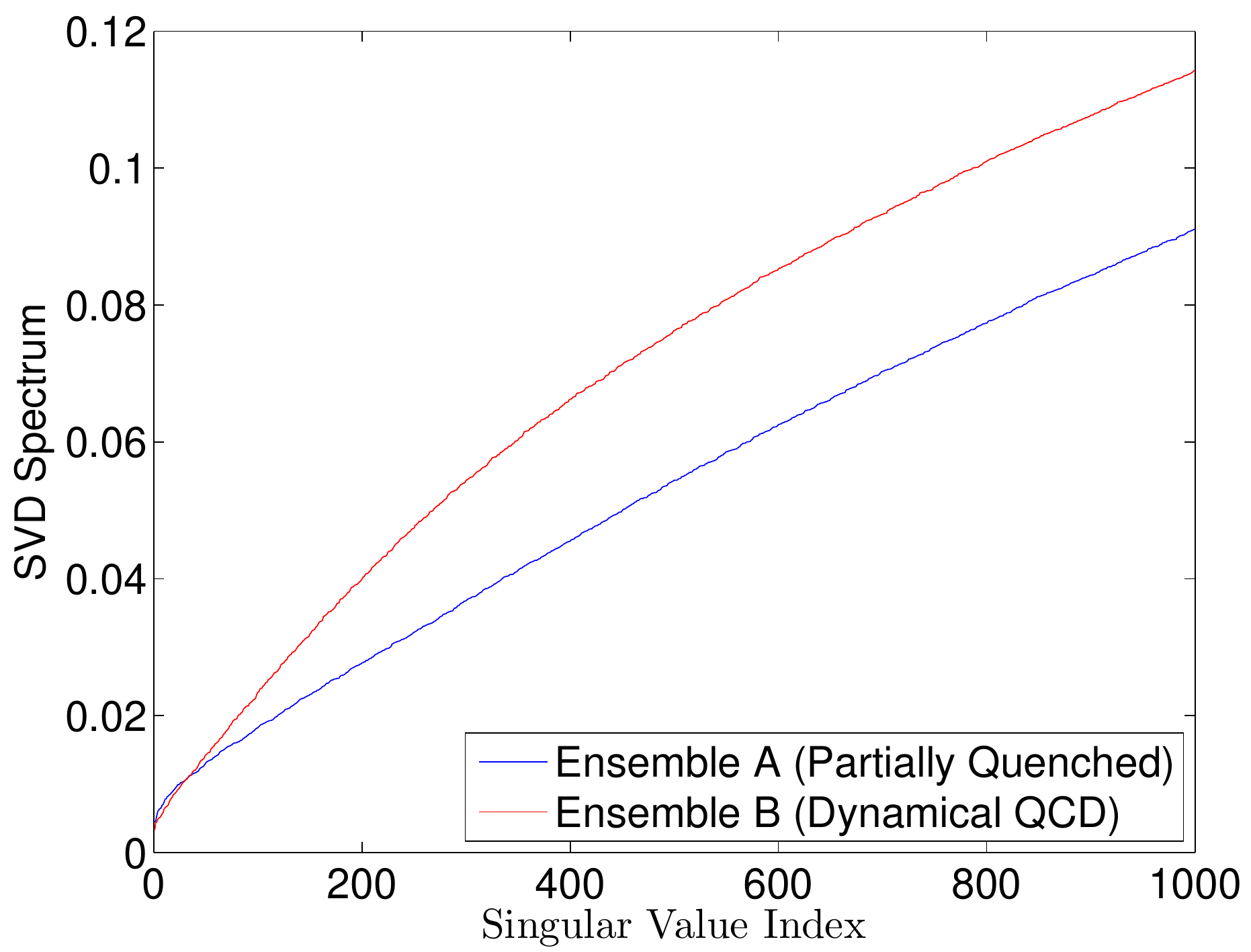}}\hspace{0em}%
  \subcaptionbox{Ensemble B $\Tr{A^{-1}}$ Deflation Model\label{QCD_plot}}{\includegraphics[width=0.48 \textwidth]{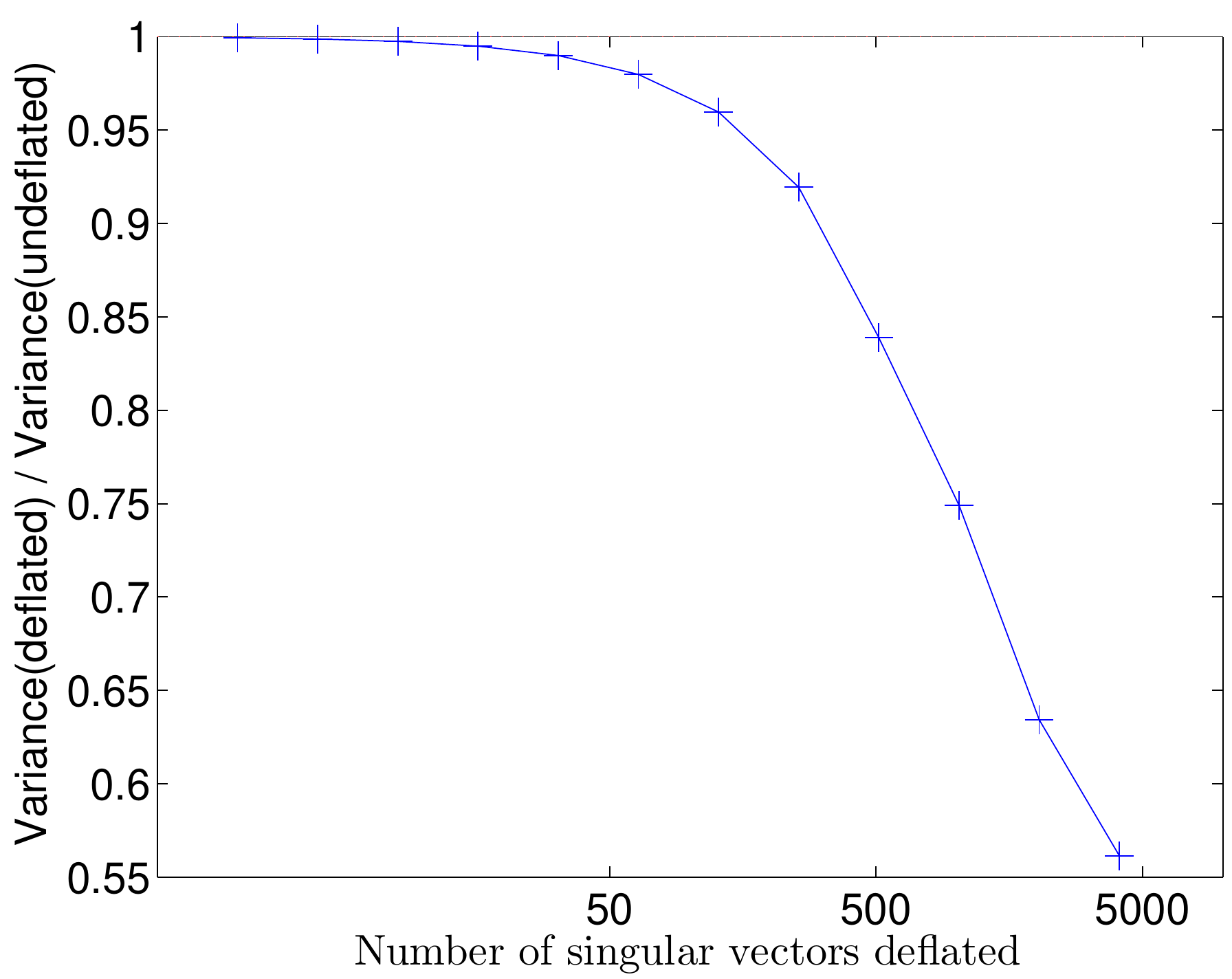}}
  \caption{The left plot displays the 1000 lowest magnitude singular values for both matrices as obtained by PRIMME. The right plot shows results from our deflation model using the 1000 computed singular values of the ensemble B matrix, 
  and simulating the rest of the spectrum as a Wilson Dirac operator 
  in free field.}
\end{figure}

\subsection{Monte Carlo with deflation}
We analyze the singular spectra of these matrices in the context of 
  our deflation theory from Section \ref{sec:theory}.
Figure \ref{SVDSpectrum} shows the smallest 1000 singular values of $A$ 
  for both ensembles.
The lowest 20 rise rapidly before the spectrum growth slows down to 
  slightly sublinear growth (ensemble B) or close to linear (ensemble A).
Since our focus is the inverse $A^{-1}$, the situation seems to similar to
  Figure \ref{2D_Laplacian_Inverse_multiN}.
To run this through our model, we wanted a rough estimate of the rest of 
  the spectrum. 
We have merged our 1000 smallest, explicitly computed singular values, with the analytically obtained singular values of the free field Wilson Dirac operator 
to obtain an approximate full spectrum for $A$. This was achieved by quadratically fitting the exactly computed singular values up to 4000 vectors and joining them with the free field spectrum via a small line segment.
Then, we use our model to simulate the effects of deflation on variance 
  for up to 5000 lowest singular triplets.
In Figure \ref{QCD_plot}, our deflation model predicts a variance reduction 
  of approximately $30\%$ for 1000 singular vectors, and $40\%$ for 
  5000 singular vectors.

Since the trace and variance of the undeflated and deflated matrix are
  not known, the model has to be compared with the statistically measured 
  variance of Monte Carlo.
An experiment with the full dynamical matrix from ensemble B was conducted to
  compute $\Tr(A^{-1})$ with the Hutchinson method using random Rademacher 
  vectors (no HP).
Table \ref{Tab:deflation} shows the results for both the undeflated
  operator (first line) and the operator deflated with a various numbers 
  of singular vectors (from 25 to 1000 starting with the smallest in 
  magnitude).
The three result columns show the statistical variance after 32, 64, and 
  128 Monte Carlo steps, respectively.
Past the lowest 25-50 singular values, there is little improvement with the 
  Monte Carlo estimator. With 128 Rademacher vectors the deflation speedup 
  is about $30\%$.
The improvement may not be impressive, but what is impressive is the 
  level of agreement with the prediction of our model in Figure \ref{QCD_plot}.
However, this agreement is not surprising since our model assumes uniformly random unitary singular vector matrices, which is approximately the case in QCD~\cite{Shuryak:1992pi,Verbaarschot:2000dy}.

\begin{table}[!ht]
\centering  
\caption{\label{Tab:deflation} Ensemble B $\Tr(A^{-1})$ Variance} 
\begin{tabular}{c c c c} 
\hline\hline 
Monte Carlo Step & 32\ & 64 \ & 128 \\ [0.5ex] 
\hline 
Undeflated & 1.0735e+04 & 5.1764e+03 & 2.7336e+03 \\
25 & 7.7396e+03 & 4.0158e+03 & 2.3081e+03 \\
50 & 7.0769e+03 & 3.8168e+03 & 2.0751e+03 \\
100 & 7.0645e+03 & 3.8108e+03 & 2.0641e+03 \\
200 & 6.9917e+03 & 3.9187e+03 & 2.1308e+03  \\ 
300 & 7.0246e+03 & 3.8921e+03 & 2.1127e+03  \\ 
400 & 6.9628e+03 & 3.9373e+03 & 2.1466e+03 \\
500 & 7.0002e+03 & 3.8166e+03 & 2.1132e+03 \\
600 & 7.1782e+03 & 3.8422e+03 & 2.0921e+03 \\
700 & 7.2679e+03 & 3.8326e+03 & 2.1068e+03 \\
800 & 7.1029e+03 & 3.8064e+03 & 2.0927e+03 \\
900 & 7.1378e+03 & 3.8768e+03 & 2.1036e+03 \\
1000 & 7.0484e+03 & 3.8355e+03 & 2.0922e+03 \\  
\hline 
\end{tabular} 
\end{table}

\subsection{Synergy between deflation and hierarchical probing}
HP used with the Hutchinson method reduces the error 
  (when run deterministically) or the variance (when run stochastically 
  as in Algorithm \ref{Alg:main}).
Depending on the conditioning of the matrix, improvements over an order of magnitude have 
  been observed \cite{Stathopoulos:2013aci}.
In Figures \ref{m250VarianceLog} and \ref{m239VarianceLog} we present
  results of Algorithm \ref{Alg:main} with the ensemble A and ensemble B matrices 
  respectively, where HP is augmented by deflation.
The error bars on the variance were estimated with the Jackknife resampling 
  procedure on 40 runs of Algorithm \ref{Alg:main} with different $z_0$ 
  noise vectors.
Local minima appear on the y axis of both plots at every power of two.
This is a characteristic of the HP method, which is meaningful only 
 at these points \cite{Stathopoulos:2013aci}.
At least one order of magnitude improvement in variance is observed with 
  deflation over HP alone.

\begin{figure}[h]
  \centering
  \subcaptionbox{Ensemble A $\Tr(A^{-1})$ Variance\label{m250VarianceLog}}{\includegraphics[width=0.49\textwidth]{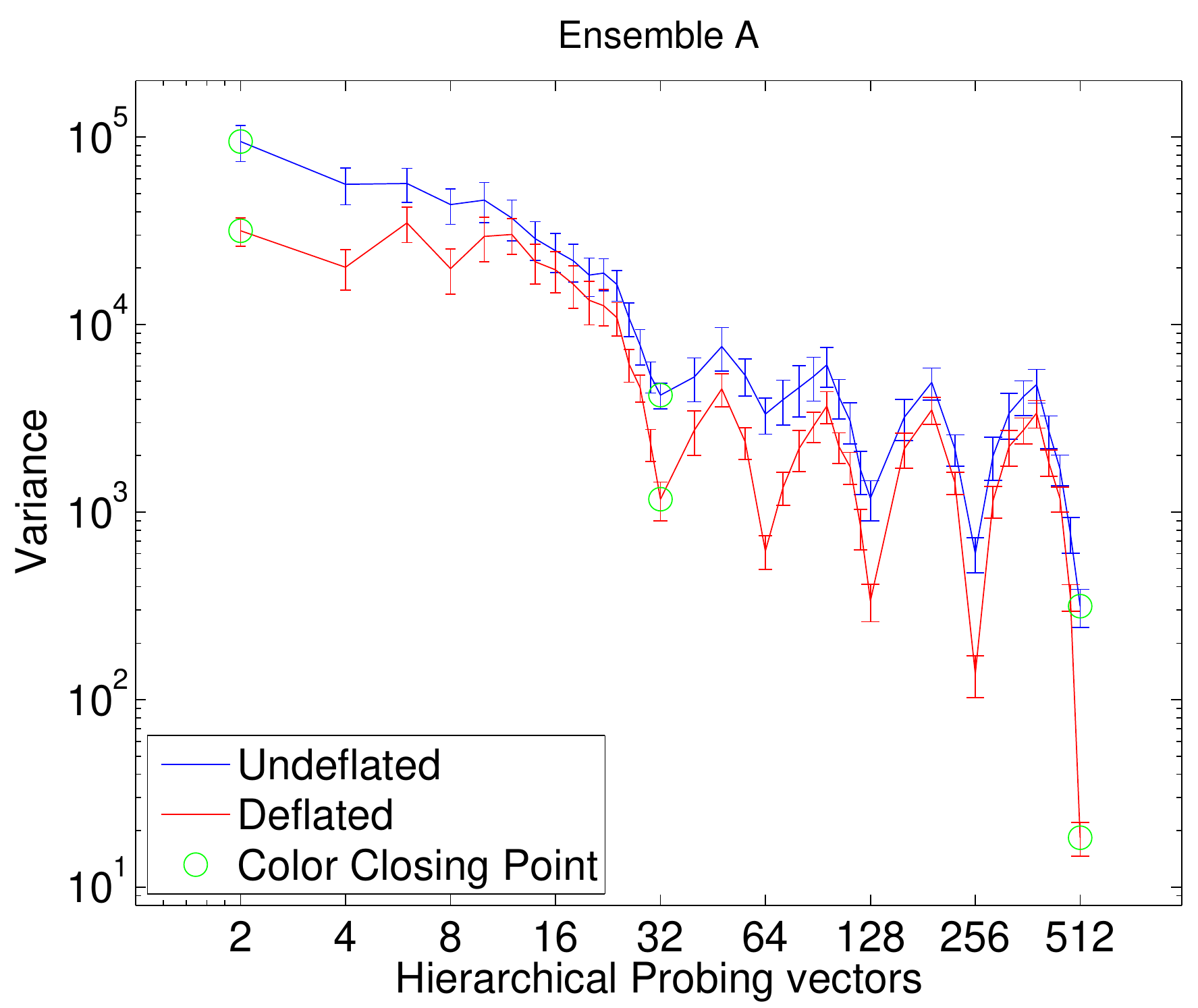}}\hspace{0em}%
  \subcaptionbox{Ensemble B $\Tr(A^{-1})$ Variance\label{m239VarianceLog}}{\includegraphics[width=0.49\textwidth]{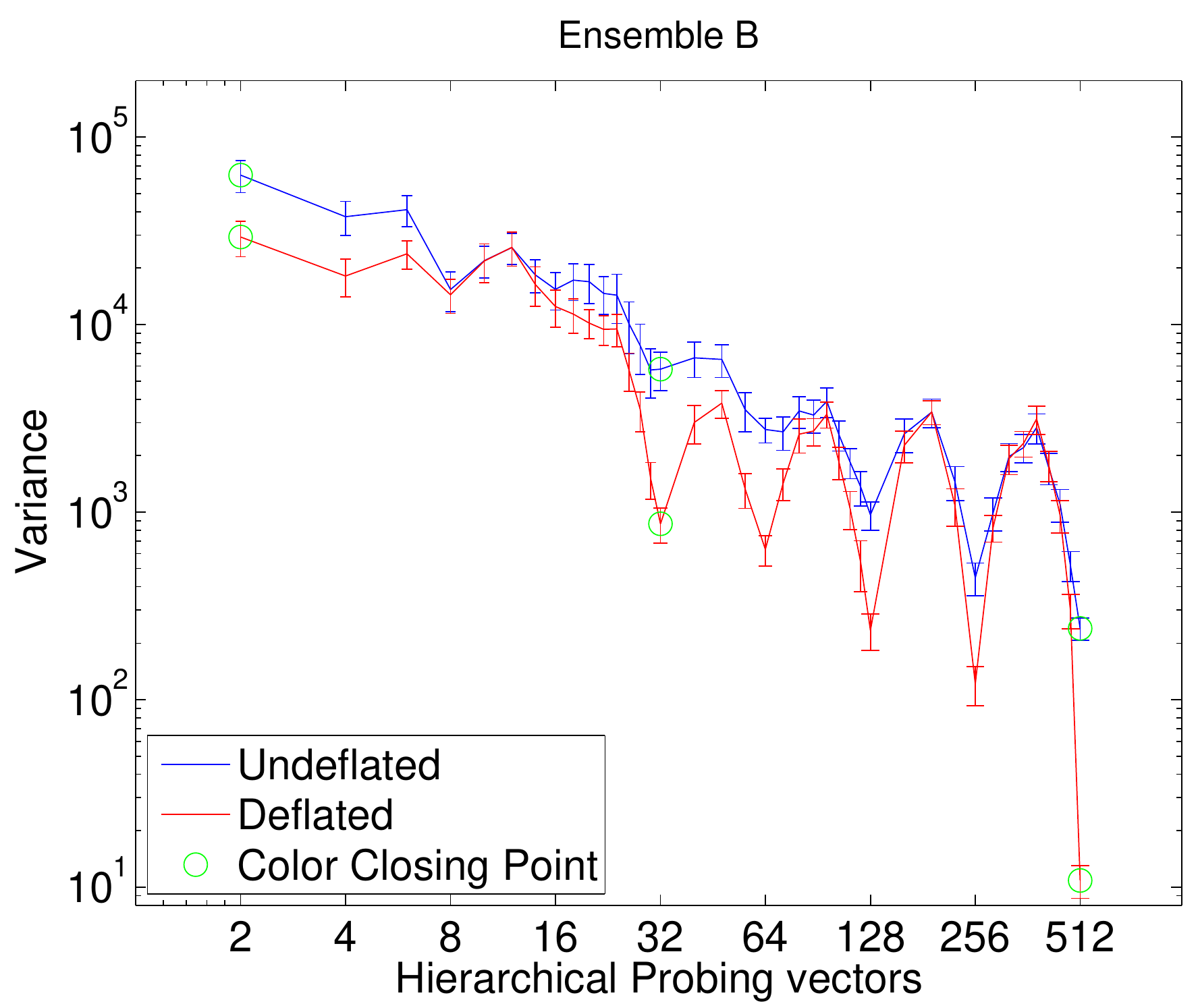}}
  \caption{Above is the variance of the hierarchical probing trace estimator with and without deflation. The full 1000 vector subspace is used as the deflated operator in red. Complete color closings are marked with green circles. For the ensemble A matrix, a factor of 15 is achieved in variance reduction between deflated and undeflated probing. Deflation yields over a factor of 20 reduction of variance for the ensemble B matrix.}
\end{figure}

\textcolor{black}{
Additionally, we compute the speedup of HP and deflated HP compared to the 
  basic MC estimator as
$$R_s=\frac{V_{stoc}}{V_{hp}(s)\times s}.$$
Here, $V_{stoc}$ is the variance from the pure noise MC estimator, 
  and $V_{hp}(s)$ is the HP variance computed with Jackknife resampling 
  over the $40$ runs.
The factor of $s$ is the number of probing vectors, and it is used to 
  normalize the speedup ratio since the error from random noise scales 
  as $(\frac{V_{stoc}}{s})^{1/2}$.
The speedup for both ensembles are displayed in figures \ref{m250SpeedUp}
  and \ref{m239SpeedUp}.
HP alone yields speedups of 2-3 instead of the speedups of 10 we noticed 
  on a matrix from Ensemble B in \cite{Stathopoulos:2013aci}.
The difference is that in the previous paper we set the quark parameter to the
  strange quark mass while in this paper we set it to the light quark mass
  which yields a much more ill conditioned matrix in Figure \ref{m239SpeedUp}.
Deflation and HP together, however, achieve a factor of 60 speedup over 
  the original Monte Carlo method.
We elaborate on this further.
}

\begin{figure}[h]
  \centering
  \subcaptionbox{Ensemble A $\Tr(A^{-1})$ Speed Up $(R_s)$\label{m250SpeedUp}}{\includegraphics[width=0.49\textwidth]{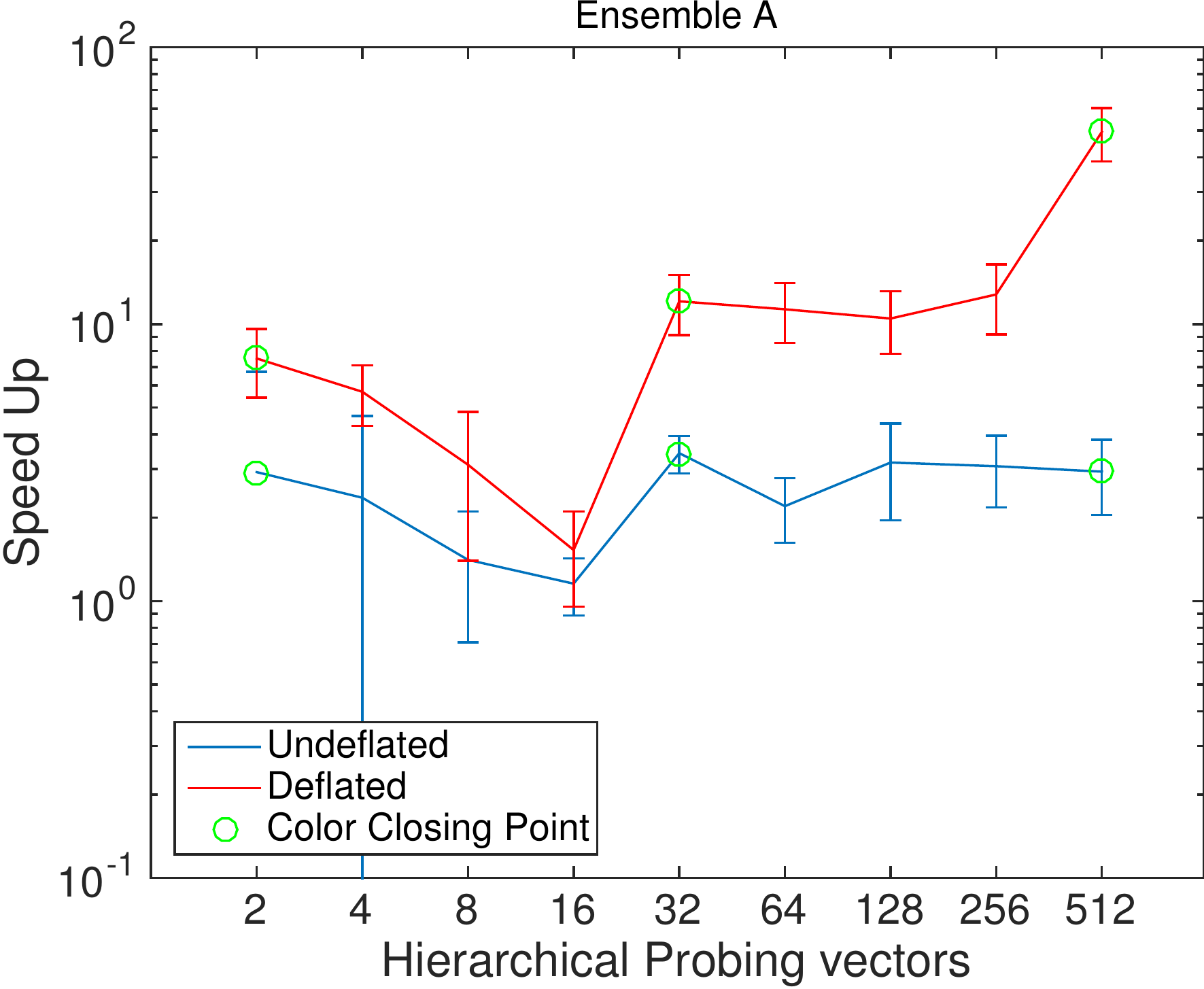}}\hspace{0em}%
  \subcaptionbox{Ensemble B $\Tr(A^{-1})$ Speed Up $(R_s)$\label{m239SpeedUp}}{\includegraphics[width=0.49\textwidth]{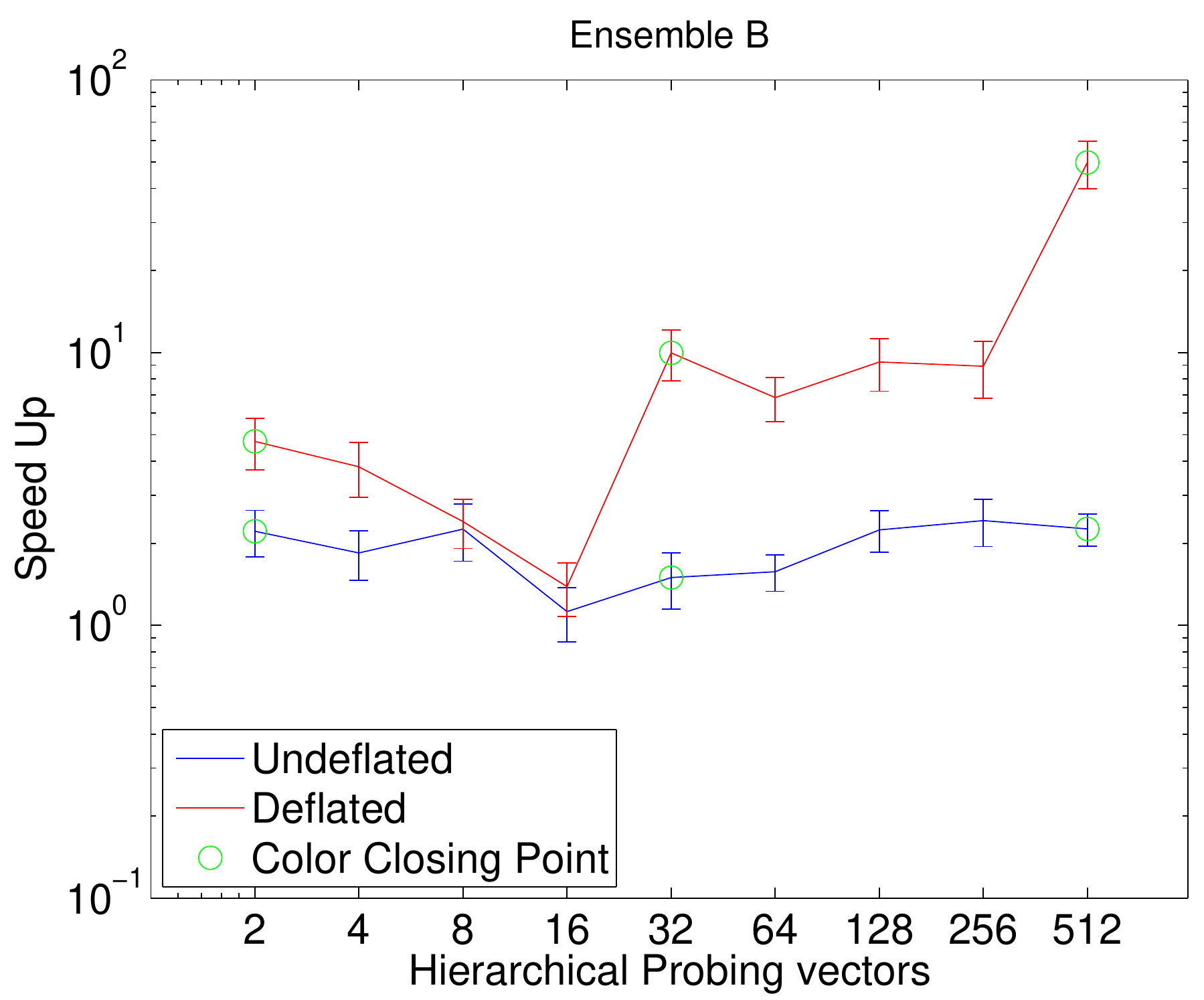}}
  \caption{
	\textcolor{black}{
  Speedup of the combined deflated HP estimator compared to pure $\mathbb{Z}_4$ noise is shown. The speedup to basic MC is estimated for both HP alone and HP with deflation. The errors are computed with Jackknife resampling.}
  }
\end{figure}

It is apparent that deflation aids the HP estimator in a much more pronounced manner 
than the basic noise estimator. 
This is because of the synergistic way deflation and HP work.
The idea of HP is based on the local decay of the Green's function. 
By assuming that the neighbors of a source node in matrix A will have 
  weights in $A^{-1}$ that decay with their distance from the source, 
  HP kills the error from progressively larger distance neighborhoods.
This works well for well conditioned matrices, but for ill conditioned ones
  the $A^{-1}$ is dominated by the contributions of the near null eigenspace.
Such contributions are typically non-local which are not captured by HP.
Deflation, however, captures exactly these contributions 
  and by removing them, a much easier structure for HP is left.
In Lattice QCD, this synergy completely resolves the scaling problem
  as the mass approaches the critical mass, 
  and significantly reduces the effects of lattice size.
  
\begin{figure}[h]
\centering
\includegraphics[width=0.6\textwidth]{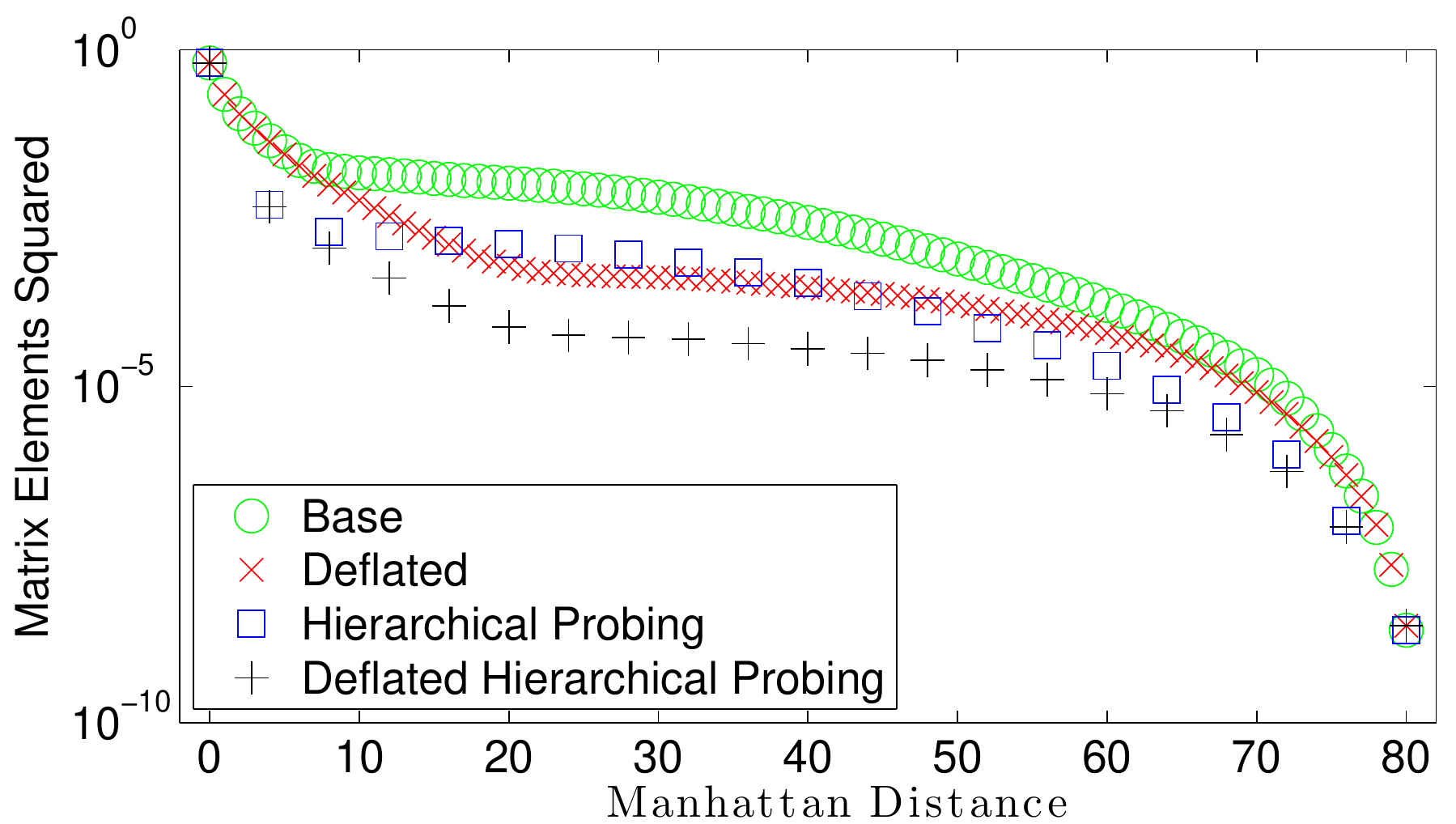}
\caption{The sum of squared absolute values of matrix elements at specified Manhattan distances from the corresponding diagonal elements for 10 randomly sampled rows.
Base case is the original Monte Carlo method.
Deflation refers to the Monte Carlo with deflation.
HP and deflated HP refer to a space spanned by the 32 hierarchical 
  probing vectors. 
A combination of HP and deflation suppresses the sum of matrix elements by orders of magnitude more than probing or deflation alone.}
\label{m250Taxi}
\end{figure}

We investigate this synergy experimentally on the matrix from ensemble A.
We seek to quantify the remaining variance on 
the original matrix ($\|A^{-1}\|_F^2$), 
after applying deflation ($\|A^{-1}_R\|_F^2$),
after applying 32 HP probing vectors $H$ ($\|(HH^H) \odot A^{-1}\|_F^2$),
and after applying both deflation and HP ($\|(HH^H) \odot A_R^{-1}\|_F^2$).
Let $B$ denote any of these four matrices. 
Since we cannot compute $\|B\|_F$ explicitly, we randomly sample 10 of 
  its rows, denoting this set as $S$. 
Then for each corresponding lattice node $i\in S$, we find all its $m_d$ 
  neighbors $j$ that are $d$ hops away in the lattice 
  (i.e., its Manhattan distance-$d$ neighborhood)
  and sum their squared absolute values $|B_{ij}|^2$.
Averaging these over all $m_d$ neighbors and all nodes in $S$ gives us 
  an estimate of how much variance remains from elements at distance $d$.
These $W_d$ are plotted in Figure \ref{m250Taxi},
$$ 
W_d = 1/|S| \sum_{i\in S} \sum_{j \in {\cal N}_d} |B_{ij}|^2/m_d, \
    \mbox{ where } {\cal N}_d = \left\{j:\ dist(i,j)=d \right\}
    \mbox{ and } m_d = |{\cal N}_d|. $$
The figure shows how HP eliminates the variance from the first
3 distances and repeats this pattern in multiples of 4 (1,2,3,5,6,7,$\ldots$)
\cite{Stathopoulos:2013aci}.
While probing eliminates better short-distance variance,
  deflation is better at long-distance.
Combining them achieves a much greater reduction in variance than either 
  of the two alone.

\subsection{Varying the SVD deflation space}

We also study the effect of the size of the deflation SVD subspace. By saving all inner products performed in the trace estimator, we are able to play back the trace simulation deflating with different numbers of singular triplets.
We combine deflation and HP and report results for 32 and 512 probing vectors,
  which represent the proper color closings for HP in a 4D lattice 
  \cite{Stathopoulos:2013aci}.
As before, the error bars are obtained from 40 different runs of
  Algorithm \ref{Alg:main} with different $z_0$.

\begin{figure}[h]
  \centering
  \subcaptionbox{Ensemble A SVD and 32 HP vectors\label{m250VarianceSVD32}}{\includegraphics[width=0.45\textwidth]{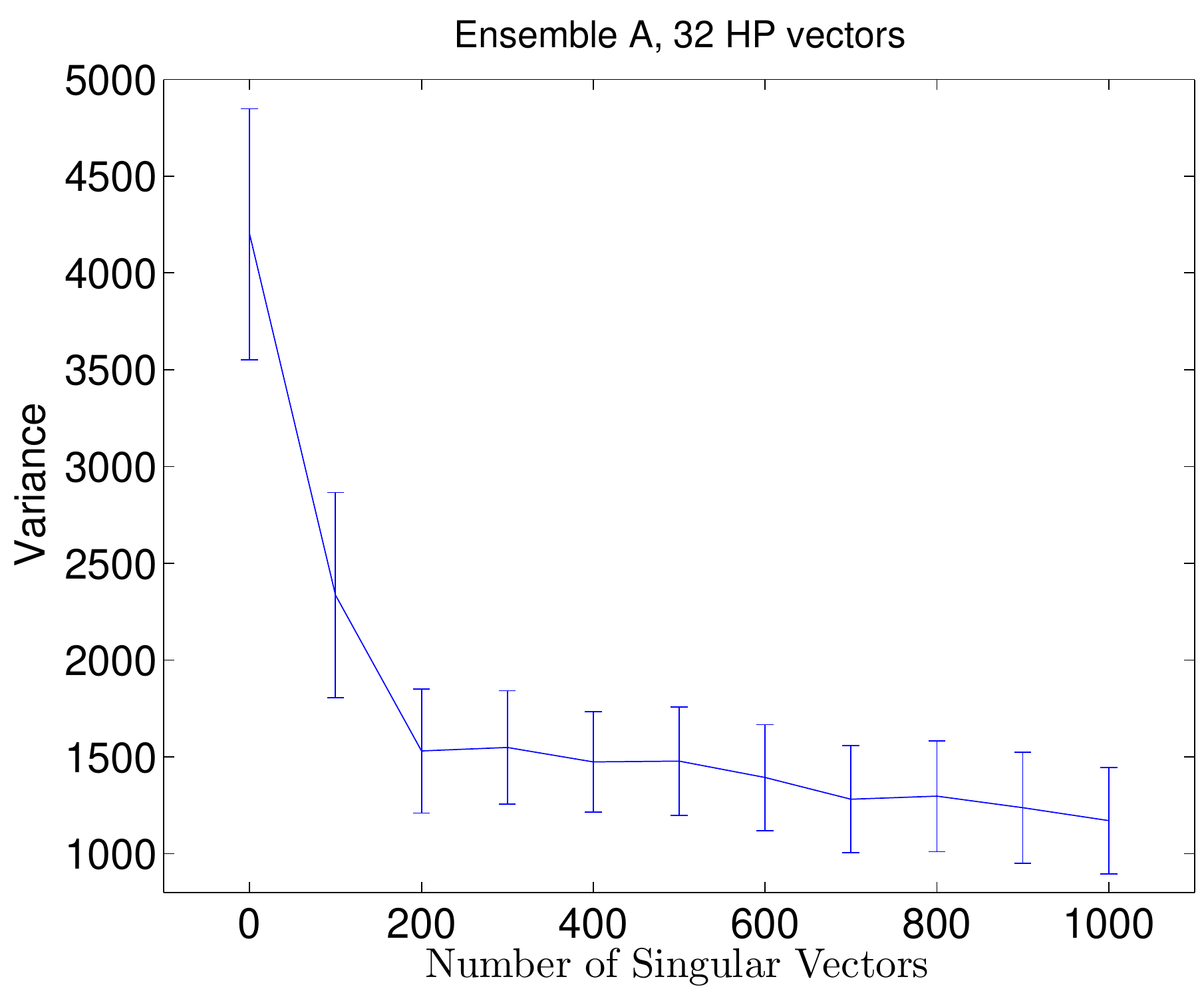}}\hspace{0em}%
  \subcaptionbox{Ensemble A SVD and 512 HP vectors\label{m250VarianceSVD512}}{\includegraphics[width=0.44\textwidth]{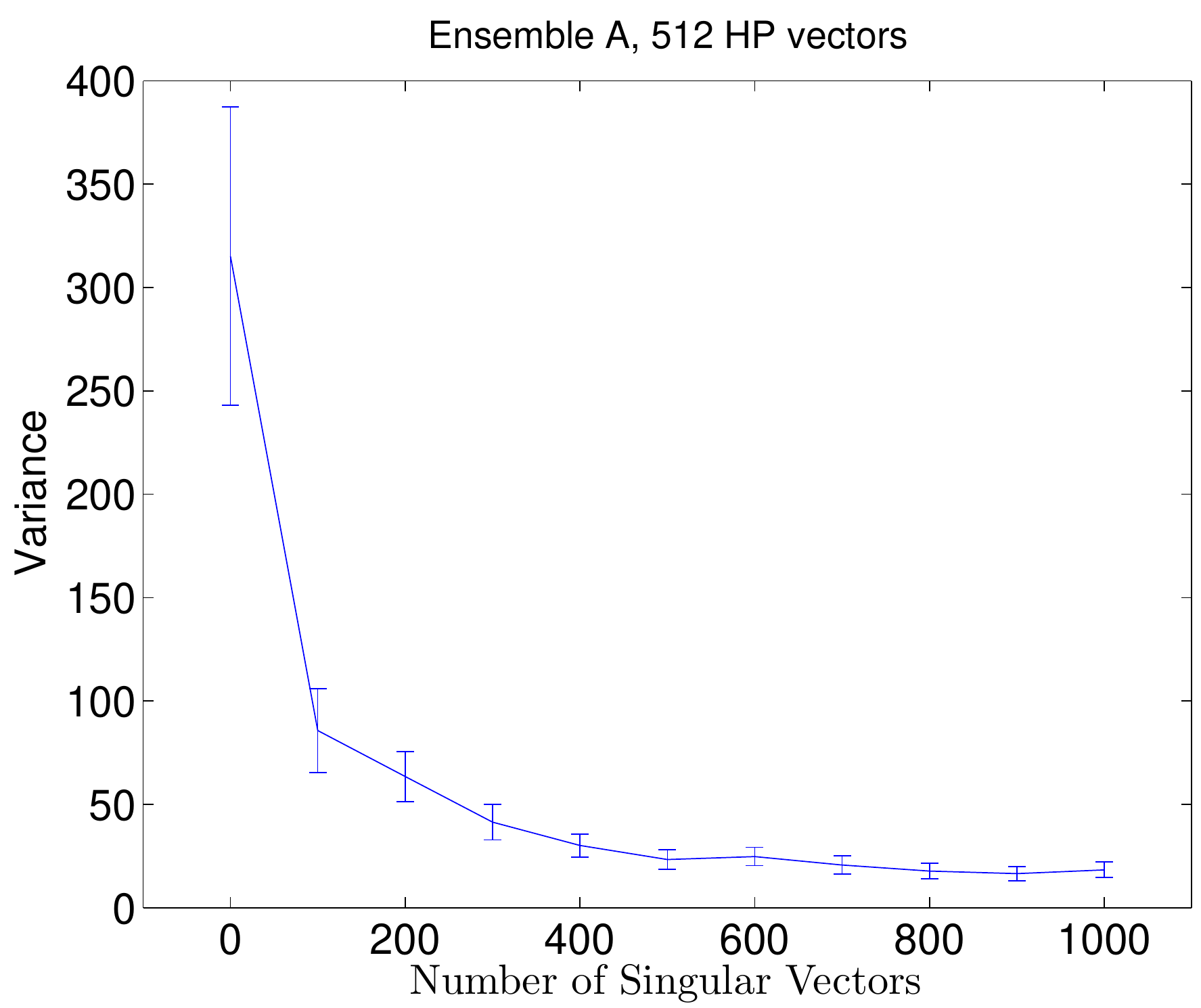}}
  \caption{Variance for the ensemble A matrix as a function of the deflated SVD subspace dimension at two color closing points of HP. The left plot is with 32 probing vectors, the right is with the full 512. }
\end{figure}
\begin{figure}[h]
  \centering
  \subcaptionbox{Ensemble B SVD and 32 HP vectors\label{m239VarianceSVD32}}{\includegraphics[width=0.45\textwidth]{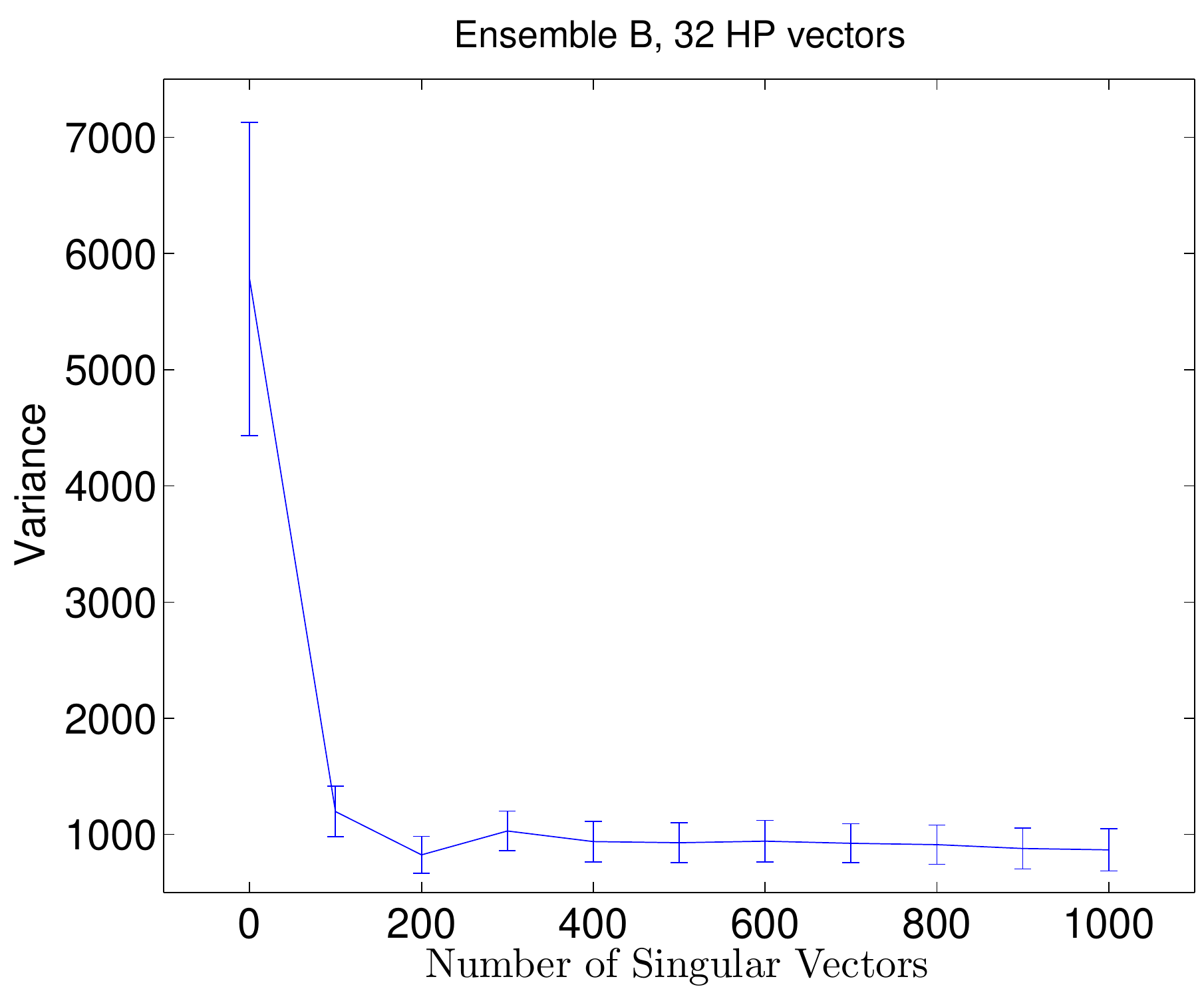}}\hspace{0em}%
  \subcaptionbox{Ensemble B SVD and 512 HP vectors\label{m239VarianceSVD512}}{\includegraphics[width=0.44\textwidth]{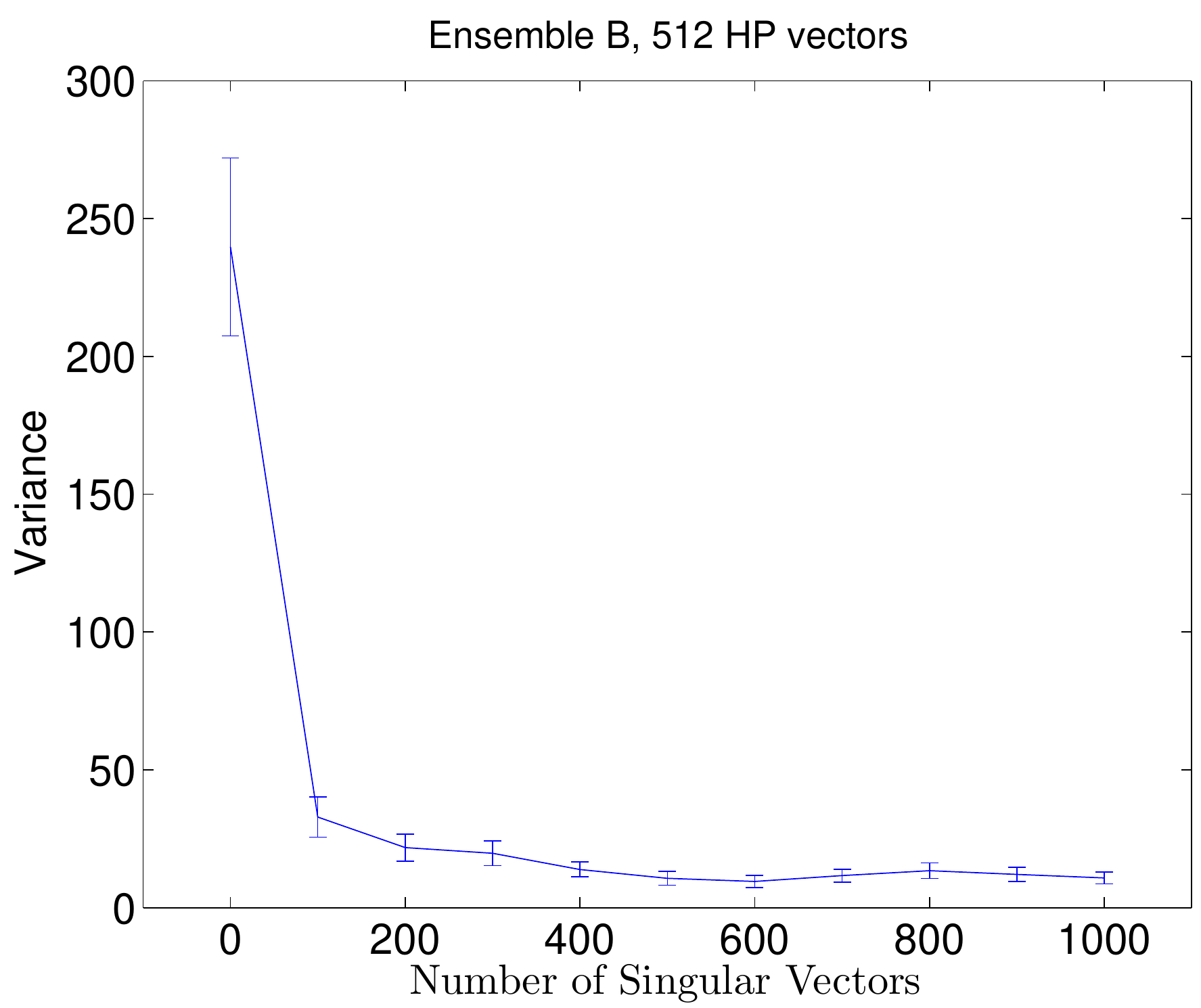}}
  \caption{Variance for the matrix from ensemble B, as a function of the deflated SVD subspace dimension at two color closing points of HP. The left plot is with 32 probing vectors, the right is with the full 512. }
\end{figure}

Figure \ref{m250VarianceSVD32} shows that deflation with 200 singular vectors 
reduces variance by a factor of 3, and beyond 200 little improvement is gained.
In Figure \ref{m250VarianceSVD512}, HP has removed the error for larger 
  distances and therefore it can use more singular vectors effectively, 
  yielding more than an order of magnitude improvement.
Still there is potential for computational savings since 500 singular vectors
  have the same effect as 1000 ones.
Figures \ref{m239VarianceSVD32} and \ref{m239VarianceSVD512} display similar attributes for the ensemble B matrix.

These experiments illustrate that the optimal number of vectors to 
  be used in each of the two techniques depends on each other. 
This is only an issue if one needs to figure out how many singular vectors to 
  compute a priori, because if these are already available, their application 
  in the method is not computationally expensive.
Moreover, while using a sufficiently large number of probing vectors 
  is important, the performance of deflation 
  seems to be much less sensitive to the number of singular vectors.
Once the near null space has been removed, there are diminishing returns
  to deflate with bigger subspaces.
In general, the effect of this can be estimated through the model 
  while computing the singular spectrum.
The experiments we provide in this paper should provide a good rule of 
  thumb when computing disconnected diagrams for a similar class 
  of Lattice QCD gauge configurations.

\subsection{Wallclock timings and efficiency}
Implementing either MC or hierarchical probing with deflation requires an additional setup cost from finding the SVD space. In Lattice QCD, this cost is of little importance since the subspace may be stored and reused several times for computing various correlation functions. 

Deflation is valuable even as a ``one shot method'' for our QCD matrices.
We investigate the case in which the trace of $A^{-1}$ only needs to be 
  computed once, and report the time to compute the SVD, first separately and 
  then as the overhead of the preprocessing of Hutchinson's method.
Our experiments were performed on the Cray Edison using 32 12-core 
  Intel Ivy Bridge nodes clocked at 2.4 GHz, each with only 8 cores 
  enabled due to memory and node topology considerations.

\begin{figure}[!ht]
\centering

\subcaptionbox{PRIMME cost\label{m250PRIMMETimings}}
{\includegraphics[width=0.45\textwidth]{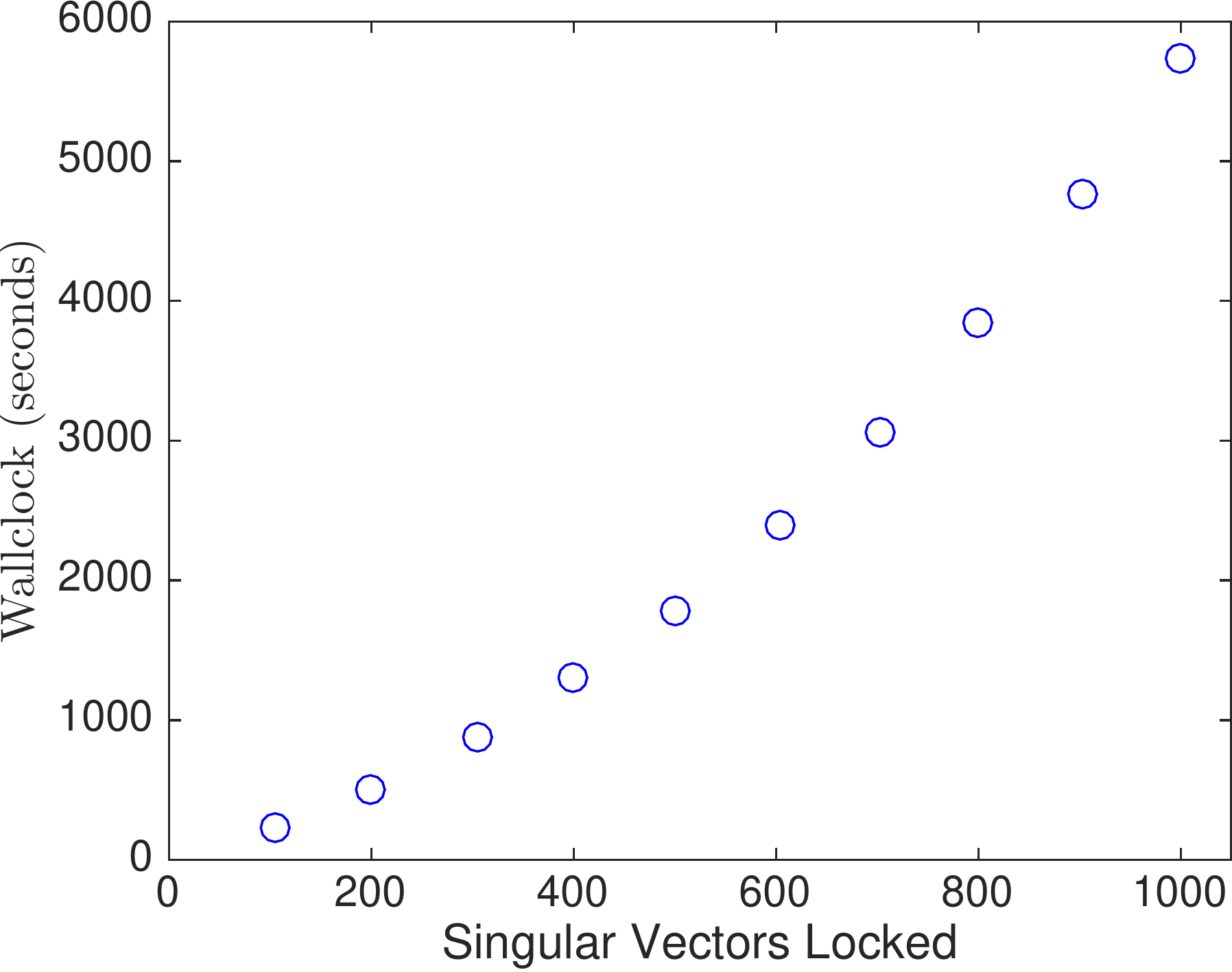}}
\subcaptionbox{
	\textcolor{black}{
	Variance vs simulation cost\label{m250TraceTimings}}
}
{\includegraphics[width=0.45\textwidth]{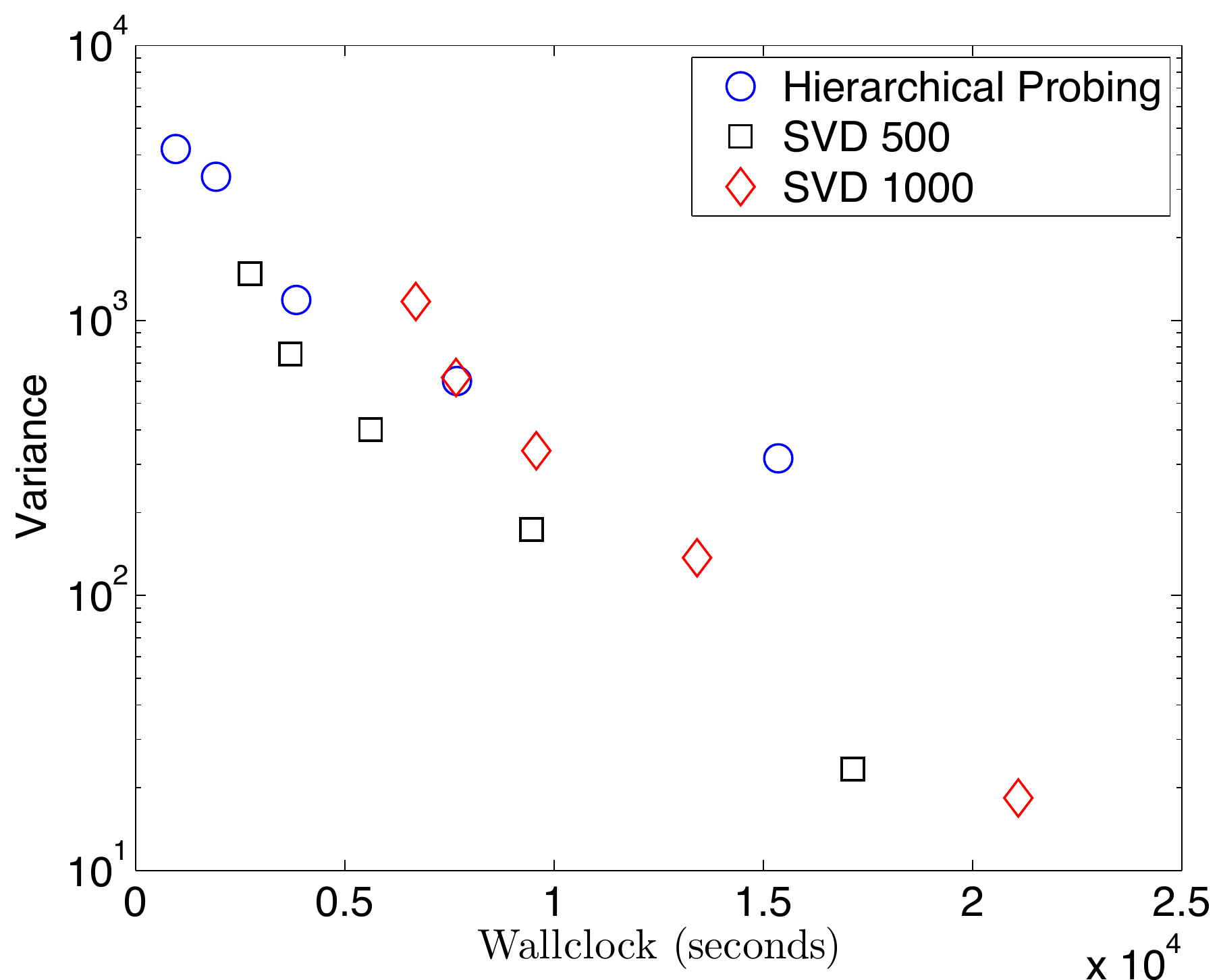}}
\caption{
Eigenvectors computed by PRIMME from 100 to 1000 for the matrix from ensemble A. 
A log plot of variance and cost. Each case displays 5 points, which represent the variance and wallclock at probing vectors 32, 64, 128, 256, and 512.
}
\end{figure}
Figure \ref{m250PRIMMETimings} shows the timings for PRIMME as a function
  of the number of eigenvectors found.
As more eigenvectors converge, orthogonalization costs increase resulting 
  in time increasing super linearly.
The expected reduction in the efficiency of the AMG preconditioner 
  as we move to the interior of the spectrum is in fact negligible. 
Obtaining 1000 eigenvectors takes 1.5 hours, while 500 vectors are computed 
  in less than half an hour.
Indeed with the help of the AMG preconditioner, PRIMME was able to solve for 
  the eigenvalues of $A^HA$ at a fraction of the cost of the probing estimator.

We now add the time to compute the singular space as well as the time 
  to perform the projections with that space to the timings for the remaining
  steps of Algorithm \ref{Alg:main}. 
We consider two simulations; one with deflation space of 500 vectors and one
  with 1000 vectors. \textcolor{black}{ From figures \ref{m239VarianceSVD32} and 
  \ref{m239VarianceSVD512} we do not expect gain beyond 500 singular triplets.}
For each closing point of HP (32, 64, 128, 256, and 512 probing vectors),
  Figure \ref{m250TraceTimings} plots the achieved variance as a function 
  of total wallclock time.
We observe that the variance with 500 deflation vectors at probing vector 128 
  is comparable to the variance of the plain HP method at 512 probing vectors.
This translates to a 3-fold reduction in wallclock, even with the SVD 
  computations included. 
Furthermore, at 512 probing and 500 deflation vectors, we see a 15-fold reduction 
  in variance with the SVD time being less than 10\% of total wallclock. 
This suggests that deflation can be used equally well as a one shot 
  method for variance reduction. 

\section{Conclusion}
We have studied theoretically and experimentally the effects of deflating 
  the near null singular value space on reducing the variance 
  of the Hutchinson method.
This is a Monte Carlo method for estimating the trace of the inverse of a
  large, sparse matrix, which among other areas is also common in Lattice QCD.
Our theoretical analysis showed that variance reduction 
\textcolor{black}{is guaranteed if the singular values of the matrix increase 
	at an exponential rate. For slower increasing rates, the singular
vector structure plays a role.}
By assuming that the singular vectors are random unitary matrices, 
  we were able to quantify the above in a concise, elegant formula
  that requires only the first two moments of the singular values.
Experiments have shown that the formulas model even general, non-random 
  matrices very well.
We have also shown an interesting property, where singular vector deflation 
  applied to Hermitian matrices can increase the variance, whereas 
  deflation applied to non-Hermitian matrices with the same spectrum 
  always decreases the variance. 

In the second part of the paper we use deflation to solve a particularly 
  challenging, large scale QCD application defined on a 4D regular lattice.
The singular values are computed using PRIMME with an AMG preconditioner 
  in one of the largest SVD computations performed in Lattice QCD.
Although deflation on its own has a limited impact on the variance, 
  combining it with the current state-of-the-art method of 
  Hierarchical Probing (HP) provides a factor of 10-15 speedup over HP.
We explain this synergy theoretically and provide a thorough experimental
  analysis that confirms our explanation. 
These Lattice QCD tests, which were performed 
  on  Edison (the Cray supercomputer at  the National Energy Research Scientific Computing Center) 
  show that our method can have significant efficiency improvements on similar Lattice QCD calculations
  that require the computation of the trace of matrices related to the inverse of the Dirac matrix.

\section*{Acknowledgments}
This work has been supported by NSF under grants No. CCF 1218349 and ACI SI2-SSE 1440700, and by DOE under a grant No. DE-FC02-12ER41890.
KO and AG have been supported by the U.S. Department of Energy through Grant Number DE- FG02-04ER41302. 
KO has been supported through contract Number DE-AC05-06OR23177 under which JSA operates the Thomas Jefferson National Accelerator Facility. 
AG has been supported by the U.S. Department of Energy, Office of Science, Office of Workforce Development for Teachers and Scientists, Office of Science Graduate Student Research (SCGSR) program. The SCGSR program is administered by the Oak Ridge Institute for Science and Education for the DOE under contract number DE-AC05-06OR23100.
This research used resources of the National Energy Research Scientific Computing Center, a DOE Office of Science User Facility supported by the Office of Science of the U.S. Department of Energy under Contract No. DE-AC02-05CH11231.

\bibliography{qcd}
\bibliographystyle{siamplain}
\end{document}